\documentclass[aps,pra,groupedaddress,twocolumn]{revtex4-2}

\usepackage{graphicx}
\usepackage{dcolumn}
\usepackage{bm}
\usepackage{hyperref}

\usepackage{amssymb}
\usepackage{amsthm}
\usepackage{mathtools}
\usepackage{amsmath}
\usepackage{placeins}
\usepackage{setspace} 
\newcommand{\Mod}[1]{\ (\mathrm{mod}\ #1)}
\usepackage{cleveref}
\usepackage{subcaption}
\usepackage[most]{tcolorbox}
\usepackage{braket}
\usepackage{multirow} 
\usepackage{float}

\makeatletter
\newcommand{\vo}{\vec{}\@ifnextchar{^}{\,}{}}
\makeatother

\newtheorem{theorem}{Theorem}
\newtheorem{definition}{Definition}

\tcbset{highlight math style={colback=white!10!white, colframe=black!10!black}}

\begin{document}

\title{\textbf{An Efficient Decomposition of the Carleman Linearized Burgers' Equation.}}

\author{Reuben Demirdjian}
\email{Reuben.Demirdjian.civ@us.navy.mil}
\affiliation{U.S. Naval Research Laboratory, Monterey, CA, 93943, United States}

\author{Thomas Hogancamp}
\affiliation{U.S. Naval Research Laboratory, Monterey, CA, 93943, United States}
\affiliation{American Society for Engineering Education, Washington, D.C.}

\author{Daniel Gunlycke}
\affiliation{U.S. Naval Research Laboratory, Washington, DC, 20375, United States}

\date{\today}

\begin{abstract}
    Herein, we present a polylogarithmic decomposition method to load the matrix from the linearized 1-dimensional Burgers' equation onto a quantum computer. First, we use the Carleman linearization method to map the nonlinear Burgers' equation into an infinite linear system of equations, which is subsequently truncated to order $\alpha$. This new finite linear system is then embedded into a larger system of equations with the key property that its matrix can be decomposed into a linear combination of $\mathcal{O}(\log n_t + \alpha^2\log n_x)$ terms for $n_t$ time steps and $n_x$ spatial grid points. While the terms in this linear combination are not unitary, each can be implemented using a simple block encoding procedure. A numerical simulation is performed by combining our approach with the variational quantuam linear solver demonstrating that accurate solutions are possible. Finally, a resource estimate shows that the upper bound of the Clifford and T gate counts scale like $\mathcal{O}(\alpha(\log n_x)^2)$ and $\mathcal{O}((\log n_x)^2)$, respectively. This is therefore the first explicit polylogarithmic data loading method with respect to $n_x$ and $n_t$ for a Carleman linearized system. 
\end{abstract}

\maketitle


\section{Introduction}
Partial differential equations (PDEs) are ubiquitous in nearly all scientific and engineering disciplines, however, their solutions are rarely analytically known. Instead, PDEs are typically solved numerically using high performance computers along with discretization methods to find approximate solutions \cite{trefethen1996finite, strikwerda2004finite, thomas2013numerical}. In computational fluid dynamics (CFD) and numerical weather prediction (NWP), the computational resources available can limit model accuracy by constraining the grid size of spatial and temporal discretizations \cite{Ferziger2002}. A spatially coarse CFD or NWP model may be unable to resolve important small-scale features of the fluid (e.g. turbulence and convection) and instead rely on parameterization or closure methods to approximate their effects, ultimately leading to error growth that can eventually corrupt the solution \cite{Pope2001, bauer2015quiet, Palmer2001}. Therefore, an increase in computational resources enables finer spatial discretizations, which may allow for fewer or more accurate parameterizations and thereby a more accurate solution \cite{Dudhia2014}.

Quantum computing is an emerging field that can exponentially speedup specific applications \cite{Montanaro2016, Babbush2023, Dalzell2023} such as solving linear systems of equations \cite{VQLS, harrow2009quantum, childs2017quantum, Huang2021}. Since CFD models rely on solving nonlinear PDEs, the Carleman linearization method has been proposed to transform the original set of nonlinear PDEs into an infinite set of linear ordinary differential equations (ODEs), which are subsequently truncated into a finite set of linear ODEs \cite{Liu21, demirdjian2022variational, li2025potential}. The advantage of this method is that a quantum linear system algorithm (QLSA) may be applied to solve the set of linear ODEs and thereby obtain an approximate solution to the original nonlinear PDE. However, there are a number of challenges that must be solved if this is to be done efficiently, and it is currently an open question whether this, or any other method of solving nonlinear PDEs, are viable on quantum computers \cite{Gaitan2020, Lubasch2020, Oz2021, Gourianov2022, Lapworth2022, Ljubomir2022, Tennie2023, Krovi2023, Song2025, Jennings2024, penuel2024feasibility, Pool2024, Jin2024, Gonzalez-Conde2024Carleman, Lewis2024, Surana2024, Gnanasekaran2023, Gnanasekaran2024ConstrainedOpt, Surana2024PolynomialDynSys, Gourianov2025}.

One such challenge is the data loading problem for linear systems of equations. This can be understood by considering the variational quantum linear solver (VQLS) \cite{VQLS} -- a variational technique used to solve linear systems of equations of the form $L\vec{x}=\vec{b}$ where $L\in\mathbb{C}^{N\times N}$ and $\vec{x},\vec{b}\in\mathbb{C}^N$. The VQLS method, among other QLSAs, relies on the linear combination of unitaries \cite{childs2012hamiltonian}, whereby $L=\sum_{l=0}^{N_s-1}c_lA_l$ for complex coefficients $c_l$ and unitary matrices $A_l\in\mathbb{C}^{N\times N}$. While any square matrix $L$ is guaranteed to have a decomposition of this form, the VQLS algorithm is only efficient if $N_s \sim \mathcal{O}(\text{poly}(\text{log} \, N))$. This restriction comes from the fact that the number of circuits in the VQLS cost function scales like $\mathcal{O}(N_s^2)$ \cite{demirdjian2022variational}. This means that $N_s$ must have a practical bound; otherwise, the quantum advantage is lost simply by executing the large number of circuits to load the linear system. Similarly, each $A_l$ circuit depth must also be bounded by $\mathcal{O}(\text{poly}(\text{log} \, N))$, otherwise quantum advantage is again lost when preparing the individual circuits. Henceforth, we refer to the problem of finding a decomposition such that both the number of circuits $N_s$ and the $A_l$ circuit depths (or a block encoding thereof) are both bounded by $\mathcal{O}(\text{poly}(\text{log} \, N))$ as the \textit{decomposition problem}. This may also be thought of as the data loading problem, which is not generally efficient and therefore requires bespoke methods for each application \cite{Gunlycke2020, Sato2021, Ali2023, GonzalezConde2024, Bae2024, Williams2024,Hogancamp2026}.

Another challenge is that the conditions under which the Carleman linearization approach is applicable is currently a topic of research. In Liu \textit{et al.}\! \cite{Liu21}, they define the parameter $R$ as the ratio of nonlinearity to dissipation finding that the Carleman linearization error converges to zero exponentially with increasing truncation order for $R<1$. Conversely, for $R \ge 1$ they find that the truncation order must grow exponentially with time to bound the error. The latter result problematic because the size of a Carleman linearized system scales exponentially with the truncation order. That being said, their analytical analysis provides only an upper bound and they even numerically show that accurate solutions are possible in the $R \ge 1$ regime using small truncation orders. Furthermore, their analysis is limited to a specific class of dissipative systems. Jennings \textit{et al.}\! \cite{Jennings2025} extended their analysis to a broader class of stable systems by defining independent $R$-numbers for each type of system. Their analysis allows for improved error bounds thereby extending the types of nonlinear dynamical systems that can be efficiently simulated with quantum computers. However, as was found in the Liu \textit{et al.}\! \cite{Liu21} analysis, it may be that the bounds imposed by analytical analyses are too conservative, in which case learning what problems can be solved may come experimentally.

\subsection{Contributions}
In this study, we solve the decomposition problem for the 1-dimensional (1D) Carleman linearized Burgers' equation with periodic boundary conditions -- a paradigmatic nonlinear PDE. Figure \ref{fig:Schematic} illustrates the methods introduced here and contrasts them with the ones used in \cite{Liu21, demirdjian2022variational}. Both the previous and proposed methods follow the same two initial steps, first the 1D Burgers' equation is discretized (box a) and then the Carleman linearization method is applied (box b). At this point, our approaches deviate. In the previous method, one would decompose the matrix $A$ (box c) and then use a QLSA (like the VQLS) to obtain a solution (box d). However, there are no known $\text{poly}(\log\,N)$ decompositions for the basic Carleman linearized Burgers' equation \cite{Gnanasekaran2024ConstrainedOpt} and therefore quantum advantage is lost. In contrast, the proposed method embeds the Carleman linearized 1D Burgers' equation into an even larger system of equations with matrix $A^{\text{(e)}}$ (box e). The benefit of this additional layer of complexity is that $A^{\text{(e)}}$ can be efficiently decomposed into $\text{poly}(\log\,N)$ terms (box f) that can be implemented into a QLSA using circuits with $\text{poly}(\log\,N)$ depths (box g). The proposed method therefore offers a quantum advantage when used in combination a QLSA that requires a linear combination of unitaries like VQLS (box d).

\begin{figure*}
  \centering
  \includegraphics[scale=0.5]{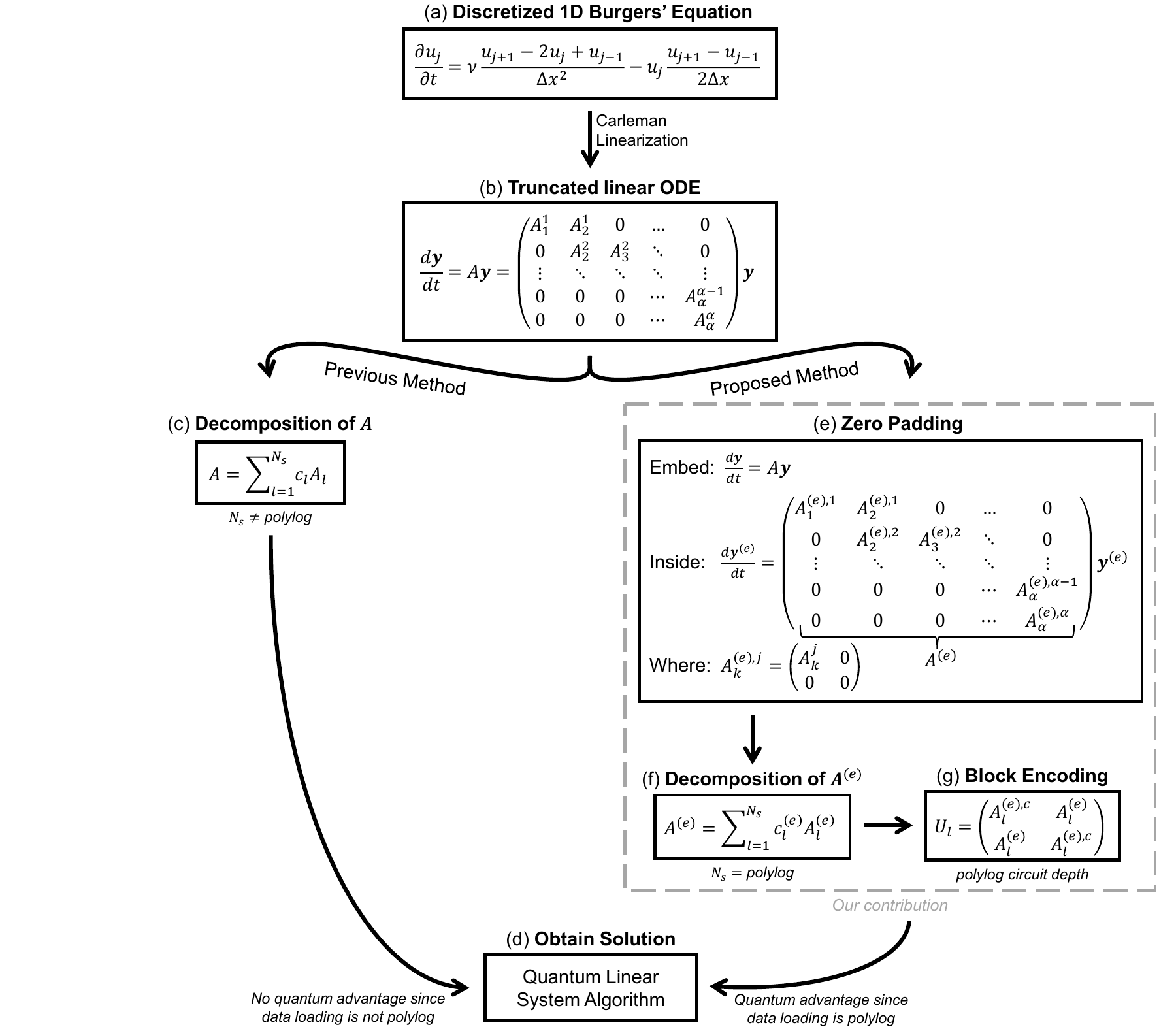}
  \caption{An illustration comparing the method proposed in this study (boxes a,b,e,f,g,d) with the previous method (boxes a,b,c,d). (a) The spatially discretized 1D Burgers' equation. (b) The truncated Carleman linearized 1D Burgers' equation. (c) Decomposition of the Carleman linearized matrix $A$ of which no known polylog decomposition exists \cite{Gnanasekaran2024ConstrainedOpt}. Note that the time discretization step is skipped in this simplification. (d) A QLSA to solve the linear system. (e) The zero padding method whereby the original system of equations $A$ are embedded in a larger system $A^{\text{(e)}}$. (f) The matrix $A^{\text{(e)}}$ is decomposed efficiently into $\text{poly}(\log\,N)$ terms. (g) Each $A_l^{\text{(e)}}$ matrix is block encoded with $\text{poly}(\log\,N)$ circuit depth.}
  \label{fig:Schematic}
\end{figure*}

The key insights presented in this paper are two-fold: (1) the zero padding method that enables us to decompose the matrix into a polylogarithmic number of terms, and (2) an extension of the block encoding method, introduced in Gnanasekaran and Surana (2024) \cite{GS24}, that enables us to efficiently block encode each term. Together, these two insights provide a polylogarithmic decomposition for the 1D Carleman linearized Burgers' equation. It is important to note that while the decomposition presented here is problem specific, we believe that the insights introduced can be generalized and applied to more complex problems. 

This work is structured as follows: In Section \ref{OverviewGS24} we present an overview of the relevant results from \cite{GS24}. The Carleman linearized 1D Burgers' equation from \cite{Liu21} is derived in Section \ref{CarlemanLin} and our zero padding approach is introduced in Section \ref{CarlemanEmbedding}. Next, in Section \ref{MatrixDecomposition} we derive an efficient decomposition for the zero padded matrix by splitting it into terms that are easily block encoded, which are subsequently presented in Section \ref{BlockEncoding}. Next, a resource estimation for our decomposition approach is performed in Section \ref{Complexity} and shown to be efficient (polylogarithmic). Our approach is then combined with the VQLS method to perform a numerical simulation in Section \ref{NumSims} demonstrating that it can be used to find accurate solutions. Finally, we present our conclusions in Section \ref{Conclusions} and discuss implications.


\section{Overview of \cite{GS24}} \label{OverviewGS24}
Define the tau basis $\mathbb{T}=\{\tau_0,\tau_1,\tau_2,\tau_3\}$ and the sigma (Pauli) basis $\mathbb{S}=\{\sigma_0,\sigma_1,\sigma_2,\sigma_3\}$ where
\begin{equation}
    \tau_0=|0\rangle\langle 0|\,, \, \tau_1=|0\rangle\langle 1|\,, \, \tau_2=|1\rangle\langle 0|\,, \, \tau_3=|1\rangle\langle 1|\,, \notag
\end{equation}
and $\sigma_0=\sigma_x,\, \sigma_1=\sigma_y, \, \sigma_2=\sigma_z, \, \text{and } \sigma_3=I$.

Suppose we have a matrix $A \in \mathbb{C}^{N \times N}, \, N=2^Q$ for some integer $Q$. In the tau basis there exists a unique decomposition $A=\sum_{l=0}^{N_\tau-1}c_lC_l$ for $N_\tau$ terms where $C_l=\bigotimes_{k=0}^{Q-1}\tau_{v_k}$, for $\tau_{v_k} \in \mathbb{T}$, $v_k\in\{0,\dots,3\}$ and $c_l \in\mathbb{C}$. Similarly, in the sigma basis there exists a unique decomposition $A=\sum_{l=0}^{N_\sigma-1}d_lD_l$ for $N_\sigma$ terms where $D_l=\bigotimes_{k=0}^{Q-1}\sigma_{w_k}$, for $\sigma_{w_k} \in \mathbb{S}$, $w_k\in\{0,\dots,3\}$ and $d_l \in\mathbb{C}$. It is important to note that, while these decompositions always exist, the number of terms ($N_\tau$ and $N_\sigma$) may be exponential for an arbitrary matrix.

To circumvent this problem, \cite{Liu2021, GS24, GS2025_LCofThings} introduced a mixed tau and sigma set given by $\mathbb{P}=\{\rho_0,\rho_1,\rho_2,\rho_3,\rho_4\}$ where $\rho_0=\tau_0,\,\rho_1=\tau_1,\,\rho_2=\tau_2,\,\rho_3=\tau_3,\,\rho_4=\sigma_3$. Using this new set, there exist non-unique decompositions of the form $A=\sum_{l=0}^{N_s-1}a_lA_l$ where $A_l=\bigotimes_{k=0}^{Q-1}\rho_{r_k}$, for $\rho_{r_k} \in \mathbb{P}$, $r_k\in\{0,\dots,4\}$ and $a_l \in\mathbb{C}$. For specific matrices, \cite{GS24} shows that there exist decompositions with $N_s=\mathcal{O}(\text{poly}(\log\,N))$, providing an exponential improvement compared to that of the tau or sigma basis alone. One challenge presented with this method, however, is that the $A_l$ matrices are not unitary. To resolve this, \cite{GS24} shows that each $A_l$ can be systematically block encoded. Furthermore, they develop a method to implement these block encodings directly into VQLS. The following constructions are adapted from Section 4 of \cite{GS24}.

\begin{definition} \label{DefCompletionComplement}
    Suppose $W\subset V$ where $V$ is a Hilbert space. For a linear operator $F:W \rightarrow V$ that preserves inner products, the unitary operator $\overline{F}:V \rightarrow V$ is called a unitary completion when $\overline{F}$ spans the whole space $V$ and $\overline{F}|w\rangle=F|w\rangle \,\,\forall\,\,|w\rangle\in W$. Additionally, $F^c\coloneq \overline{F}-F$ is the unitary complement of $F$ and is unique for a specific choice of $\overline{F}$ given $F$.
\end{definition}
Note that, while \cite{GS24} uses the term unitary complement, we opted for the more general and widely used term unitary complement. Also, the unitary completion always exists and is not necessarily unique (see Def. 2 of \cite{GS24} and Ex. 2.67 of \cite{Mike&Ike}). Following Definition \ref{DefCompletionComplement}, Theorem 2 of \cite{GS24} describes how to construct $\overline{A}_l$ for decompositions in $\mathbb{P}$. If $A_l = \bigotimes_k\rho_{r_k}$, then $\overline{A}_l = \bigotimes_k\overline{\rho}_{r_k}$ where
\begin{equation} \label{eqn:rhoBar}
    \overline{\rho}_{r_k} =
    \begin{cases} 
        \sigma_0, & \rho_{r_k} \in \{\rho_1, \rho_2\} \\
        \sigma_3, & \rho_{r_k} \in \{\rho_0, \rho_3,\rho_4\} 
   \end{cases} \,.
\end{equation}

Therefore, each $A_l$ may be block encoded with an associated unitary matrix $U_l \in \mathbb{C}^{2N \times 2N}$ by
\begin{equation} \notag
    U_l = \begin{pmatrix} A_l^c & A_l \\ A_l & A_l^c \end{pmatrix} \,,
\end{equation}
where we have followed the convention of \cite{GS24} by encoding $A_l$ in the upper right block instead of the upper left. Furthermore, Theorem 3 of \cite{GS24} shows that $U_l$ can be implemented using at most $Q=\log N$ single qubit gates and a single $C^qX$ gate where $q\le Q$. Finally, they derive efficient quantum circuits to calculate the local VQLS cost function based on this block encoding strategy.


\section{Carleman Linearization} \label{CarlemanLin}
The 1D Burgers' equation with periodic boundary conditions and domain length $L_x$ is given by
\begin{equation*}
\begin{gathered}
    \frac{\partial u}{\partial t}=\nu\frac{\partial^2 u}{\partial x^2} - u\frac{\partial u}{\partial x} \,, 
    \\[4pt]
    u(x,0)=u^0(x) \,,\quad u(0,t)=u(L_x,t) \,,
\end{gathered}
\end{equation*}
where $u(x,t)$ is the fluid velocity, $\nu$ is the diffusion coefficient. This can be discretized into
\begin{equation} \label{eqn:discreteBurgers}
\begin{gathered} 
	    \frac{\partial u_j}{\partial t} 
	    = \frac{\nu}{\Delta x^2}(u_{j+1} - 2u_j + u_{j-1}) \\
	    \qquad - \frac{u_j}{2\Delta x}(u_{j+1}-u_{j-1}) \,, 
	    \\[4pt]
    u_j(0) = u_j^0 \,, \quad u_0(t)=u_{n_x}(t) \,,
\end{gathered}
\end{equation}
where $\Delta x$ is the grid spacing and $\vec{u}=(u_0,\dots,u_{n_x-1})^T$ is the fluid velocity at each grid point. This can be rewritten in the form
\begin{equation} \notag
    \frac{\partial \vec{u}}{\partial t} = F_1 \vec{u} + F_2 \vec{u}^{\,\otimes 2}
    , \quad
    \vec{u}(0) = \vec{u}^{\,0} \,,
\end{equation}
where $F_1 \in \mathbb{C}^{n_x \times n_x}$ and $F_2 \in \mathbb{C}^{n_x \times n_x^2}$. Following \cite{Liu21}, the Carleman linearized 1D Burgers' equation with truncation order $\alpha=2^r$ for integer $r$ takes the form
\begin{equation}
\begin{gathered} \label{eqn:Carl}
    \frac{d\vec{y}}{dt} 
    = A\vec{y} \,, \\[4pt]
    \vec{y}(0) = ((\vec{u}^{\,0}),(\vec{u}^{\,0})^{\otimes 2}, \dots, (\vec{u}^{\,0})^{\otimes \alpha})^T \,,
\end{gathered}
\end{equation}
where
\begin{equation} \label{eqn:Adef}
    A \coloneq
    \begin{pmatrix}
        A_1^1 & A_2^1 & 0 & ... & 0\\
        0 & A_2^2 & A_3^2 & ...  & 0 \\
        \vdots & \dots & \ddots & \ddots & A_\alpha^{\alpha-1} \\
        0 & 0 & \dots & 0 & A_\alpha^\alpha
    \end{pmatrix} \,,
\end{equation}
and $\vec{y}=(\vec{u},\vec{u}^{\,\otimes 2},\dots,\vec{u}^{\,\otimes \alpha})^T \in \mathbb{C}^{\Delta}$, $\Delta=\sum_{j=1}^\alpha n_x^j$ and $A \in \mathbb{C}^{\Delta \times \Delta}$. Furthermore, $A_j^j\in\mathbb{C}^{n_x^j\times n_x^j}$ and $A_{j+1}^j\in\mathbb{C}^{n_x^j\times n_x^{j+1}}$ are defined as
\begin{subequations}
\label{equations}
\begin{align}
    \label{eqn:Ajj}
    A_j^j &\coloneq \sum_{l=0}^{j-1} I_{n_x}^{\otimes l} \otimes F_1 \otimes I_{n_x}^{\otimes j-l-1} , \\
    \label{eqn:Ajp1j}
    A_{j+1}^j &\coloneq \sum_{l=0}^{j-1} I_{n_x}^{\otimes l} \otimes F_2 \otimes I_{n_x}^{\otimes j-l-1} \,,
\end{align}
\end{subequations}
where $I_n\coloneq I^{\otimes \log n}$. Using the backward Euler discretization with $n_t$ time steps, \eqref{eqn:Carl} becomes
\begin{equation} \label{eqn:LYB}
    L\vec{Y} = \vec{B} \,,
\end{equation}
which is expanded into 
\begin{align} \notag
    &\begin{pmatrix}
        I  & 0 & \dots & 0 \\
        -I & M & \dots & 0 \\
        \vdots & \ddots & \ddots & \vdots \\
        0  & \dots & -I & M
    \end{pmatrix}
    \begin{pmatrix}
        \vec{y}^{\,0} \\
        \vec{y}^{\,1} \\
        \vdots \\
        \vec{y}^{\,n_t-1}
    \end{pmatrix}
    =
    \begin{pmatrix}
        \vec{y}^{\,0} \\
        \vec{0}_\Delta \\
        \vdots \\
        \vec{0}_\Delta
    \end{pmatrix} \,,
\end{align}
where $M=I-\Delta t A$, $L \in \mathbb{C}^{n_t\Delta \times n_t\Delta}$, $\vec{Y},\vec{B}\in\mathbb{C}^{n_t\Delta}$, $\vec{0}_\Delta$ is the zero vector of size $\Delta$, and $\vec{y}^{\,m}=\vec{y}(m\Delta t)$.


\section{Zero Padding} \label{CarlemanEmbedding}
Following the approach of \cite{GS24,Liu2021}, one would attempt to write the matrix in \eqref{eqn:Adef} as a linear combination of elements from $\mathbb{P}$. The sparsity and highly patterned structure of $A$ suggests that this can be done efficiently. However, the non-square $A^j_{j+1}$ terms create a serious technical impediment. To overcome this challenge, we embed \eqref{eqn:LYB} into a larger system in which judicious zero padding creates a convenient square block structure. First, we define $A_j^{\text{(e)},j}\in \mathbb{C}^{n_x^\alpha \times n_x^\alpha}$ by embedding the associated lower dimensional $A_j^j$ matrices given by
\begin{equation} \label{eqn:Aejj}
\begin{split}
    &A_j^{\text{(e)},j} \coloneq
    \rho_0^{\otimes \log n_x^{\alpha-j}} \otimes A_j^j \\[4pt]
    &=\begin{pmatrix}
        A_j^j & 0_{n_x^j \times (n_x^\alpha-n_x^j)} \\
        0_{(n_x^\alpha-n_x^j) \times n_x^j} & 0_{(n_x^\alpha-n_x^j)\times(n_x^\alpha-n_x^j)}
    \end{pmatrix} \,,
\end{split}
\end{equation}
where $j\in\{1,\dots,\alpha\}$ and $n_x=2^s$ for an integer $s$. Similarly, we define $A_{j+1}^{\text{(e)},j}\in \mathbb{C}^{n_x^\alpha \times n_x^\alpha}$  terms by embedding the $A_{j+1}^{j}$ matrix given by
\begin{equation} \label{eqn:Aejp1j}
\begin{split}
    A_{j+1}^{\text{(e)},j} &\coloneq
    \begin{pmatrix}
        A_{j+1}^j & 0_{n_x^j \times (n_x^\alpha-n_x^{j+1})} \\
        0_{(n_x^\alpha-n_x^j)\times n_x^{j+1}} & 0_{(n_x^\alpha-n_x^j)\times(n_x^{\alpha}-n_x^{j+1})}
    \end{pmatrix} \\[4pt]
    &= \rho_0^{\otimes\log(n_x^{\alpha-j-1})} \\
    &\quad\otimes
    \sum_{l=0}^{j-1}
    \Biggl[
    \Bigl( \rho_0^{\otimes \log n_x} \otimes K^{(n_x^l,n_x)} \Bigr) \\
    &\quad\qquad\cdot 
    \Biggl(
     \begin{pmatrix}
        F_2 \\
        0_{(n_x^2-n_x)\times n_x^2}
    \end{pmatrix} 
    \otimes I_{n_x}^{\otimes l} 
    \Biggr)\cdot K^{(n_x^2,n_x^l)} \Biggr] \\
    &\quad\otimes I_{n_x}^{\otimes j-l-1} \,,
\end{split}
\end{equation}
where $K^{(a,b)} \in \mathbb{C}^{(ab\times ab)}$ denotes the commutation matrix. See Appendix \ref{DerivationAejp1j} for \eqref{eqn:Aejp1j}'s full derivation.

We can now define an analogous version of the matrix $A$ in \eqref{eqn:Adef} given by
\begin{equation} \label{eqn:Ae}
\begin{split}
    A^{\text{(e)}}
    &\coloneq
    \begin{pmatrix}
        A_1^{\text{(e)},1} & A_2^{\text{(e)},1} & 0 & ... & 0\\
        0 & A_2^{\text{(e)},2} & A_3^{\text{(e)},2} & ...  & 0 \\
        \vdots & \dots & \ddots & \ddots & A_\alpha^{\text{(e)},\alpha-1} \\
        0 & 0 & \dots & 0 & A_\alpha^{\text{(e)},\alpha}
    \end{pmatrix} \\[4pt]
     &=
    \sum_{j=1}^\alpha (\rho_{f(b_\alpha(j-1),b_\alpha(j-1))} \otimes A_j^{\text{(e)},j}) \\
    &\quad+ \sum_{j=1}^{\alpha-1} (\rho_{f(b_\alpha(j-1),b_\alpha(j))} \otimes A_{j+1}^{\text{(e)},j}) \,,
\end{split}
\end{equation}
where $A^{\text{(e)}} \in \mathbb{C}^{\alpha n_x^\alpha \times \alpha n_x^\alpha}$. Following \cite{Gunlycke2020}, the function $f:\{0,1\}^K\times\{0,1\}^K\to \{0,1,2,3\}^K$ is defined as $f(i_K,j_K)=f_{K-1}\dots f_0$ where each quaternary bit is calculated by $f_k=2i_k+j_k$ for $i_K\coloneq i_{K-1}\dots i_0$, $j_K\coloneq j_{K-1}\dots j_0$ and $k=0,\dots,K-1$. The function $b_\beta(j)$ maps the base-ten number $j\in\{1,\dots,\alpha\}$ to a binary number with $\log\beta$ digits with $\beta=2^Q$ for some integer $Q$. Together, these functions are used to map row and column decimal indices into the quaternary bitstring $f_{K-1}\dots f_0$, allowing for the convenient shorthand notation: $\rho_{f_{K-1}}\otimes \dots\otimes\rho_{f_0}$. For clarity, Appendix \ref{QuaternaryMapping} shows several examples.

We may now define the embedded system of equations, analogous to \eqref{eqn:LYB}, as
\begin{equation} \label{eqn:LeYeBe}
    L^{\text{(e)}}\vec{Y}^{\text{(e)}} = \vec{B}^{\text{(e)}} \,, 
\end{equation}
where
\begin{equation}\label{eqn:FullBurgCarlSys}\
    L^{\text{(e)}} =
    \begin{pmatrix}
        I  & 0 & \dots & 0 \\
        -I & M^{\text{(e)}} & \dots & 0 \\
        \vdots & \ddots & \ddots & \vdots \\
        0  & \dots & -I & M^{\text{(e)}}
    \end{pmatrix} \,,
\end{equation}
and $L^{\text{(e)}}\in\mathbb{C}^{\alpha n_tn_x^\alpha \times \alpha n_tn_x^\alpha}$, $\vec{Y}^{\text{(e)}}=(\vec{y}^{\,\text{(e)},0},\dots\vec{y}^{\,\text{(e)},n_t-1})^T$, $\vec{B}^{\text{(e)}}=(\vec{y}^{\,\text{(e)},0},\vec{0}_{\alpha (n_t-1)n_x^\alpha})^T$, $M^{\text{(e)}}=I-\Delta t A^{\text{(e)}}$, and $\vec{y}^{\,\text{(e)},m}=((\vec{u}^{\,m}),\vec{z}_1,(\vec{u}^{\,m})^{\otimes 2},\vec{z}_2,\dots,(\vec{u}^{\,m})^{\otimes\alpha})^T$ for the $m^\text{th}$ time step and $\vec{z}_j\in\mathbb{C}^{n_x^\alpha-n_x^j}$. Note that the structure of the $A^{\text{(e)}}$ matrix will force the $\vec{z}_j$-vectors to be zero. The process used to obtain \eqref{eqn:LeYeBe} from \eqref{eqn:LYB}  is a specific zero padding approach to embed the original Carleman linearized system into a larger dimensional system of equations. In this case, the zero padded system has a polylogarithmic decomposition as will be shown.


\section{Decomposition of $L^{\textnormal{(e)}}$} \label{MatrixDecomposition}
We now demonstrate how to decompose $L^{\text{(e)}}$ from \eqref{eqn:FullBurgCarlSys} into a linear combination of terms of the form 
\begin{equation} \label{eqn:LCNU}
    L^{\text{(e)}}=\sum_{l=0}^{N_s-1}c_lL_l
\end{equation}
where $c_l\in\mathbb{C}$ and the terms $L_l\in\mathbb{C}^{\alpha n_tn_x^\alpha \times \alpha n_tn_x^\alpha}$ are tensor products of certain well-known unitary matrices with elements in $\mathbb{P}$. First, separate the identity blocks by
\begin{equation} \label{eqn:Le}
    L^{\text{(e)}} = L_1^{\text{(e)}} - \Delta t L_2^{\text{(e)}} \,,
\end{equation}
where
\begin{equation} \notag
\begin{split}
    L_1^{\text{(e)}} &= 
    \begin{pmatrix}
        I  & 0 & \dots & 0 \\
        -I & I & \dots & 0 \\
        \vdots & \ddots & \ddots & \vdots \\
        0  & \dots & -I & I        
    \end{pmatrix} \,,
\end{split}
\end{equation}
and
\begin{equation} \notag
\begin{split}
    L_2^{\text{(e)}} &= 
    \begin{pmatrix}
        0  & 0 & \dots & 0 \\
        0 & A^{\text{(e)}} & \dots & 0 \\
        \vdots & \ddots & \ddots & \vdots \\
        0  & \dots & 0 & A^{\text{(e)}}
    \end{pmatrix} \,.
\end{split}
\end{equation}
Following \cite{GS24}, $L_1^{\text{(e)}}$ can be split into just $\log n_t+1$ terms provided by 
\begin{equation} \label{eqn:L1e}
\begin{split}
    L_1^{\text{(e)}} &= \bigg(\rho_4^{\otimes \log n_t} - \rho_4^{\otimes (\log(n_t)-1)}\otimes \rho_2 \\
    &\quad- \sum_{j=2}^{\log n_t} \rho_4^{\otimes (j-2)} \otimes \rho_2 \otimes \rho_1^{\otimes (\log(n_t)-j+1)}\bigg) \\
    &\quad \otimes \rho_4^{\otimes\log(\alpha n_x^\alpha)} \,,
\end{split}
\end{equation}
where $n_t=2^m$ for an integer $m$. Next, we split $L_2^{\text{(e)}}$ by
\begin{equation} \label{eqn:L2e}\
    L_2^{\text{(e)}} = \rho_4^{\otimes\log n_t} \otimes A^{\text{(e)}} - \rho_0^{\otimes \log n_t} \otimes A^{\text{(e)}} \,.
\end{equation}
By evaluating \eqref{eqn:Ae} into \eqref{eqn:L2e} we obtain
\begin{equation} \label{eqn:L2aL2b}
    L_2^{\text{(e)}} = L_{2a}^{\text{(e)}} + L_{2b}^{\text{(e)}} \,,
\end{equation}
where 
\begin{subequations}
    \begin{equation}\begin{split}\label{eqn:L2a}
        L_{2a}^{\text{(e)}} &= 
        \bigg(\big(\rho_4^{\otimes\log n_t} - \rho_0^{\otimes \log n_t}\big) \\
        &\otimes
        \sum_{j=1}^\alpha \big(\rho_{f(b_\alpha(j-1),b_\alpha(j-1))} \otimes A_j^{\text{(e)},j}\big)\bigg) \,,
    \end{split}\end{equation}
        \begin{equation}\begin{split}\label{eqn:L2b}
        L_{2b}^{\text{(e)}} &=
        \bigg(\big(\rho_4^{\otimes\log n_t} - \rho_0^{\otimes \log n_t}\big) \\
        &\otimes
        \sum_{j=1}^{\alpha-1} \big(\rho_{f(b_\alpha(j-1),b_\alpha(j))} \otimes A_{j+1}^{\text{(e)},j}\big)\bigg) \,.    
    \end{split}\end{equation}
\end{subequations}
$L_2^{\text{(e)}}$ has two types of terms: (1) the $L_{2a}^{\text{(e)}}$ terms associated with $A_j^{\text{(e)},j}$, and (2) the $L_{2b}^{\text{(e)}}$ terms associated with $A_{j+1}^{\text{(e)},j}$. We handle their decompositions separately in the next two subsections.

\subsection{Decomposition of $L_{2a}^{\textnormal{(e)}}$} \label{Decomp_Aejj}
First, by inserting \eqref{eqn:Ajj} into \eqref{eqn:Aejj} it can be seen that the $A_j^{\text{(e)},j}$ decomposition depends upon $F_1$. Conveniently, the $F_1$ term can be decomposed into $2\log n_x+3$ elements of $\mathbb{P}$ for the case of periodic boundary conditions provided by
\begin{equation} \label{eqn:A11}
\begin{split}
    \frac{\Delta x ^2}{\nu}F_1 &= -2\rho_4^{\otimes s} + \rho_4^{\otimes (s-1)}\otimes (\rho_1 + \rho_2) \\
    &\quad + \rho_1^{\otimes s} + \rho_2^{\otimes s} +
    \sum_{i=2}^s \rho_4^{\otimes (i-2)} \\
    &\quad \otimes\bigg(\rho_2\otimes\rho_1^{\otimes (s-i+1)} + \rho_1\otimes\rho_2^{\otimes (s-i+1)}\bigg) \,,
\end{split}
\end{equation}
where $n_x=2^s$. By inserting \eqref{eqn:A11}, \eqref{eqn:Ajj}, and \eqref{eqn:Aejj}, into \eqref{eqn:L2a} we can see that $L_{2a}^{\text{(e)}}$ is decomposed into a linear combination of purely elements from $\mathbb{P}$. It is therefore straightforward to calculate the VQLS cost function using the methods in \cite{GS24}. 

\subsection{Decomposition of $L_{2b}^{\textnormal{(e)}}$} \label{Decomp_Aejp1j}
Next, we look at the $L_{2b}^{\text{(e)}}$ terms. From \eqref{eqn:Aejp1j} it is plain to see that $A_{j+1}^{\text{(e)},j}$ depends upon the matrix $\begin{pmatrix} F_2 \\ 0_{(n_x^2-n_x)\times n_x^2} \end{pmatrix}$. To gain some insight into the structure of this matrix, we consider the case for $n_x=4$ shown in Appendix \ref{Example_n4}. In general, the nonzero elements exist only in the first $n_x$-rows, and each of these rows has exactly two nonzero elements. We can therefore split this matrix into two terms given by
\begin{equation} \label{eqn:A21}
    \begin{pmatrix} F_2 \\ 0_{(n_x^2-n_x)\times n_x^2} \end{pmatrix}=-(F^+ - F^-)/(2\Delta x) \,,
\end{equation}
where $F^+$ contains the $u_ju_{j+1}$ terms and $F^-$ contains the $u_ju_{j-1}$ terms from \eqref{eqn:discreteBurgers}. These matrices can be decomposed into products of a diagonal matrix and a permutation matrix by 
\begin{equation} \notag
    F^+ = \mathcal{D}P^+ , \;\;\; F^- = \mathcal{D}P^- \,,
\end{equation}
where $\mathcal{D},P^+,P^- \in \mathbb{C}^{n_x^2 \times n_x^2}$. The $\mathcal{D}$-matrix is defined by
\begin{equation} \label{eqn:D1}
\begin{split}
    \mathcal{D} &\coloneq 
    \begin{pmatrix}
        \rho_4^{\otimes \log n_x} & 0_{n_x \times (n_x^2-n_x)} \\
        0_{(n_x^2-n_x) \times n_x} & 0_{(n_x^2-n_x) \times (n_x^2-n_x)}
    \end{pmatrix} \\[4pt]
    &= \rho_0^{\otimes \log n_x} \otimes \rho_4^{\otimes \log n_x} \,.
\end{split}
\end{equation}
The permutation matrices $P^+$ and $P^-$ are not unique since they may be written as
\begin{equation} \notag
    P^+ = 
    \begin{pmatrix} F_2^+ \\ (F_2^+)^c \end{pmatrix}
    , \quad
    P^- = 
    \begin{pmatrix} F_2^- \\ (F_2^-)^c \end{pmatrix} \,,
\end{equation}
where $F_2^+, F_2^-\in \mathbb{C}^{n_x \times n_x^2}$ are the unique positive and negative element positions of $F_2$ respectively (see Appendix \ref{Example_n4}), and $(F_2^+)^c, (F_2^-)^c\in \mathbb{C}^{(n_x^2-n_x) \times n_x^2}$ are their unitary complements, which are not unique by Definition \ref{DefCompletionComplement}. As shown in Appendix \ref{Permutation_Matrices}, there exists a choice for $(F_2^+)^c$ and $(F_2^-)^c$ such that $P^+=P_2^+P_1$ and $P^-=P_2^-P_1$ for known $P_1$, $P_2^+$ and $P_2^-$. Their associated quantum circuits are 
\begin{subequations}\label{equations}
\begin{align}
    \label{eqn:P1}
    P_1 &= \prod_{q=0}^{s-1} CX(s-q-1,2s-q-1) \,, \\
    \label{eqn:P2p}
    P_2^+ &= X_0 \, CX(0,1) \left( \prod_{q=0}^{s-3} C^{q+2}X(0,\dots,q+2) \right) \,, \\    
    \label{eqn:P2m}
    P_2^- &= \left( \prod_{q=0}^{s-3} C^{a}X(0,\dots,a) \right) CX(0,1) \, X_0 \,,
\end{align}
\end{subequations}
where $a=s-q-1$ and $n_x=2^s$. Here, $C^jX(q_0,\dots,q_j)$ is a multi-control NOT gate whereby the first $q_0,\dots,q_{j-1}$ arguments are control qubits and the final $q_j$ argument is the target. Additionally, $CX(q_{j-1},q_j)$ is the CNOT gate with control on the $q_{j-1}$ qubit and target on the $q_j$ qubit, and $X_0$ is the NOT-gate applied to the $0^{\text{th}}$ qubit. Note that the complexity of the $P_2^+$ and $P_2^-$ matrices can be improved upon as discussed in \cite{Incrementer}.

The final component of \eqref{eqn:Aejp1j} to decompose is the commutation matrix, which is given by
\begin{equation} \label{eqn:com_mat}
    K^{(a,b)} = \prod_{r=0}^{n-1} \prod_{q=0}^{m-1} S(r+m-q-1,r+m-q) \,,
\end{equation}
where $a=2^m$, $b=2^n$, $S(i,j)$ is the SWAP gate between the $i^{\text{th}}$ and $j^{\text{th}}$ qubits. The circuit depth complexities for $P_1$, $P_2^-$, $P_2^+$, and $K^{(a,b)}$ are all polylogarithmic and are discussed in Section \ref{Complexity}.


\section{Block Encoding} \label{BlockEncoding}

The work in Section \ref{MatrixDecomposition} provides us with a linear combination of non-unitary matrices for $L^{\text{(e)}}$. The next step towards generalizing the technique of \cite{GS24} requires us to block encode each term of this linear combination. If the general form of the original linear combination is given in \eqref{eqn:LCNU}, then we must block encode each $L_l$ into a unitary matrix $U_l\in \mathbb{C}^{2\alpha n_tn_x^\alpha \times 2\alpha n_tn_x^\alpha}$ where $U_l=U_{l,1}U_{l,2}$ as discussed in \cite{GS24}. For convenience, in the remainder of this section the subscript $l$ is dropped. 

As discussed in Section \ref{MatrixDecomposition}, $L^{\text{(e)}}$ is split into three types of terms $L_1^{\text{(e)}}$, $L_{2a}^{\text{(e)}}$ and $L_{2b}^{\text{(e)}}$. Since both the $L_1^{\text{(e)}}$ and $L_{2a}^{\text{(e)}}$ terms were shown in Section \ref{MatrixDecomposition} to be decomposed into purely elements from $\mathbb{P}$, they can be treated following \cite{GS24}. In contrast, the $L_{2b}^{\text{(e)}}$ terms are decomposed into products of elements from $\mathbb{P}$ with the permutation matrices introduced in Section \ref{Decomp_Aejp1j}. The remainder of this section will focus on demonstrating that the methods in \cite{GS24} can be extended to block encode the $L_{2b}^{\text{(e)}}$ terms.

By evaluating \eqref{eqn:Aejp1j} into \eqref{eqn:L2b}, we can see that the $L_{2b}^{\text{(e)}}$ terms have the general form
\begin{equation} \label{eqn:GeneralAejp1j}
\begin{split}
    \mathcal{A}&=\Bigl( \bigotimes_{k=0}^{Q_1-1} \rho_{r_k} \Bigr) \\
    &\quad\otimes 
    \Biggl(
    \Bigl( \rho_0^{\otimes \log n_x} \otimes K^{(n_x^l,n_x)} \Bigr) \\
    &\quad\cdot \Bigl(\mathcal{D}P \otimes I_{n_x}^{\otimes l} 
    \Bigr)
    \cdot K^{(n_x^2,n_x^l)} \Biggr)
    \otimes I_{n_x}^{\otimes j-l-1} \,,
\end{split}
\end{equation}
where $\mathcal{A}\in\mathbb{C}^{\alpha n_tn_x^\alpha \times \alpha n_tn_x^\alpha}$, $\rho_{r_k}\in \mathbb{P}$, $\,r_k\in\{0,\dots,4\}$, $\,P \in \{P^+,P^-\}$, $\mathcal{D}$ is defined in \eqref{eqn:D1}, and $Q_1=\log(\alpha n_tn_x^\alpha/n_x^{j+1})$.

\begin{theorem} \label{ACompletionTheorem}
One choice of unitary completion for \eqref{eqn:GeneralAejp1j} is given by 
\begin{equation} \label{eqn:ABar}
\begin{split}
    \overline{\mathcal{A}}&=\Bigl( \bigotimes_{k=0}^{Q_1-1} \overline{\rho}_{r_k} \Bigr) \\
    &\quad\otimes 
    \Biggl(
    \Bigl( \rho_4^{\otimes \log n_x} \otimes K^{(n_x^l,n_x)} \Bigr)\\
    &\quad\cdot \Bigl(P \otimes I_{n_x}^{\otimes l} \Bigr) \cdot K^{(n_x^2,n_x^l)} \Biggr)
    \otimes I_{n_x}^{\otimes j-l-1} \,,
\end{split}
\end{equation} 
where $\overline{\rho}_k$ is defined in \eqref{eqn:rhoBar}. 
\end{theorem}
\begin{proof}
    See Appendix \ref{ACompletionProof}.
\end{proof}

Theorem \ref{ACompletionTheorem} shows that a simple procedure exists for each unitary completion of the $L_{2b}^{\text{(e)}}$ terms. Next, using this result we show that matrices of the form \eqref{eqn:GeneralAejp1j} have a simple block encoding.
\begin{theorem} \label{U1U2}
For a matrix $\mathcal{A}$ as defined in \eqref{eqn:GeneralAejp1j}, the following relations are true:
\begin{equation} \notag
    U \coloneq 
    \begin{pmatrix}
        \mathcal{A}^c & \mathcal{A} \\
        \mathcal{A} & \mathcal{A}^c
    \end{pmatrix}
    = U_1U_2 \,,
\end{equation}
where we have
\begin{equation} \notag
\begin{split}
    U_1 &\coloneq
    \begin{pmatrix}
        I - \mathcal{A}\mathcal{A}^T & \mathcal{A}\mathcal{A}^T \\
        \mathcal{A}\mathcal{A}^T & I - \mathcal{A}\mathcal{A}^T
    \end{pmatrix} \,, \\
    U_2 &\coloneq
    \begin{pmatrix}
        \overline{\mathcal{A}} & 0 \\
        0 & \overline{\mathcal{A}} \\
    \end{pmatrix} \,.
\end{split}
\end{equation}
Moreover, both $U_1$ and $U_2$ are unitary matrices. 
\end{theorem}

\begin{proof}
    See Appendix \ref{U1U2Proof}.
\end{proof}

Theorem \ref{U1U2} demonstrates that the block encoded matrix $U$ can be implemented by two simpler unitary operations. The following two theorems show that each of these unitary operations have have polylogarithmic gate-depths.
\begin{theorem}
For $\mathcal{A}$ defined as in \eqref{eqn:GeneralAejp1j}, the $U_1$ matrix given by Theorem \ref{U1U2} can be implemented with a single $C^qX$ gate, where $q\le \log(\alpha n_tn_x^\alpha)$.
\end{theorem}
\begin{proof}
    First, observe that $\rho_{r_k}\rho_{r_k}^T\in\{\rho_0,\rho_3,\rho_4\}$ for $\rho_{r_k}\in\mathbb{P}$. Using this property and by evaluating \eqref{eqn:D1} into \eqref{eqn:AAT}, it follows that  $\mathcal{A}\mathcal{A}^T$ is composed solely of terms from the set $\{\rho_0,\rho_3,\rho_4\}$. Thus, $\mathcal{A}\mathcal{A}^T$ is a binary diagonal matrix exactly as in Theorem 3 of \cite{GS24} and, therefore, their proof that $U_1$ can be implemented with a single multi-control gate is also applicable here. Following \cite{GS24}, the upper bound on $q$ simply comes from the number of qubits required to implement $\mathcal{A}$, which in this case is $\log(\alpha n_tn_x^\alpha)$.
\end{proof}
The explicit circuit implementation of the $U_1$ matrix is given in the proof of Theorem 3 in \cite{GS24}. An important result of theirs is that the number of control qubits is equal to the number of $\rho_{r_k}\rho_{r_k}^T\in\mathbb{P}\backslash \{\rho_4\}$ terms in the $\mathcal{A}\mathcal{A}^T$ expansion. Next, we show that the $U_2$ circuit is also efficient.

\begin{theorem} \label{U2}
For $\mathcal{A}$ defined as in \eqref{eqn:GeneralAejp1j}, the $U_2$ matrix given by Theorem \ref{U1U2} can be implemented with gate depth equal to the combined depths of $P$, $K^{(n_x^l,n_x)}$, and $K^{(n_x^2,n_x^l)}$ plus at most $\log(\alpha n_tn_x^{\alpha-2})$ Pauli-$X$ gates.
\end{theorem}
\begin{proof}
    \noindent From the definition of $U_2$ and Theorem \ref{ACompletionTheorem} we have, 
    \begin{equation} \notag
    \begin{split}
        U_2 &= \begin{pmatrix} \overline{\mathcal{A}} & 0 \\ 0 & \overline{\mathcal{A}} \end{pmatrix}
        = \rho_4 \otimes \overline{\mathcal{A}} \\
        &= \rho_4 \otimes
        \Bigl( \bigotimes_{k=0}^{Q_1-1} \overline{\rho}_{r_k} \Bigr) \\
        &\quad\otimes 
        \Biggl(
        \Bigl( \rho_4^{\otimes \log n_x} \otimes K^{(n_x^l,n_x)} \Bigr) \\
        &\quad\cdot \Bigl(P \otimes I_{n_x}^{\otimes l} \Bigr) \cdot K^{(n_x^2,n_x^l)} \Biggr)
        \otimes I_{n_x}^{\otimes j-l-1} \,.
    \end{split}
    \end{equation}
    So the $U_2$ complexity depends upon $P$, $K^{(n_x^l,n_x)}$, and $K^{(n_x^2,n_x^l)}$. Additionally, there are $Q_1$ tensor products of $\overline{\rho}_{r_k}$ terms. From Theorem \ref{ACompletionTheorem} it follows that $Q_1=\log(\alpha n_tn_x^\alpha/n_x^{j+1})$, and since $\overline{\rho}_{r_k}\in\{I,X\}$, then there are at most $\log(\alpha n_tn_x^\alpha/n_x^{j+1})$ Pauli-X gates. This is maximized for $j=1$, so we have at most $\log(\alpha n_tn_x^{\alpha-2})$ Pauli-X gates. 
\end{proof}

Theorems \ref{ACompletionTheorem} \-- \ref{U2} demonstrate how to block encode the $L_{2b}^{\text{(e)}}$ terms with complexity that depends on the $P$, $K^{(n_x^l,n_x)}$, and $K^{(n_x^2,n_x^l)}$ matrices. In the next section, we show that the gate-depth complexities for the circuit implementations of these matrices is polylogarithmic, thereby demonstrating that $U_2$ is efficient.


\section{Resource Estimation} \label{Complexity}

\begin{table*}
	\begin{center}
		\renewcommand{\arraystretch}{1.5}
		\begin{tabular}{m{0.085\linewidth}| m{0.08\linewidth} | m{0.17\linewidth} | m{0.065\linewidth} | m{0.115\linewidth} | m{0.22\linewidth}}
			\hline
			\multicolumn{6}{c}{Cost to Encode Carleman Linearized 1D Burgers' Equation} \\		
			\hline
			Term & 
			\multicolumn{2}{c|}{Clifford Count (Min, Max)} & 
			\multicolumn{2}{c|}{T Count (Min, Max)} & 
			No. of Terms (Exact) \\
			\hline\hline 
			\multirow{1}{0em}{$L_1^{\text{(e)}}$} & 
			$1$ &
			$\log n_t$ &
			$0$ &
			$\log n_t$ &			
			$\log n_t + 1$ \\
			\hline
			\multirow{1}{0em}{$L^{\text{(e)}}_{2a}$} &
			$\log \alpha$ & 
			$\log \alpha n_t n_x^\alpha$ &
			$\log \alpha$ &
			$\log \alpha n_t n_x^\alpha$ &
			$\alpha(\alpha+1)(2\log n_x +3)$ \\
			\hline
			$L^{\text{(e)}}_{2b}$ &
			$\log \alpha n_x$ &
			$\alpha(\log n_x)^2$ &
			$\log \alpha n_x$&
			$(\log n_x)^2$ &
			$2\alpha(\alpha-1)$ \\
			\hline
		\end{tabular}
	\end{center}
	\caption{A resource estimate of the Clifford and T-gate counts, and number of terms for the zero padded Carleman linearized 1D Burgers' equation from \eqref{eqn:LeYeBe}. The $L_1^{\text{(e)}}$ term comes from \eqref{eqn:L1e}, the $L_{2a}^{\text{(e)}}$ from \eqref{eqn:L2a}, and the $L_{2b}^{\text{(e)}}$ from \eqref{eqn:L2b}. Since each term is split into many terms of varying cost, we provide the minimum and maximum Clifford and T-gate counts for each to provide a range. The Clifford and T-gate count columns use big O notation, though the $\mathcal{O}(\cdot)$ has been dropped for convenience. The Number of Terms column is an exact count.}
	\label{tbl: Res Est No QLSA}
\end{table*}

There are three types quantities relevant for resource estimation \cite{Scholten2024}: (1) the number of qubits, (2) the number of terms in the linear combination of $L^{\text{(e)}}$ ($N_s$ from \eqref{eqn:LCNU}), and (3) the Clifford (composed of Hadamard, phase, and CNOT gates) and T (aka $Z^{1/4}$) gate counts required to load each matrix in the decomposition. The number of qubits required is exactly $\log (\alpha n_t n_x^\alpha) + 1$ where the additional qubit is an ancillary qubit required for the block encoding strategy outlined in Section \ref{BlockEncoding}. The latter two estimates are summarized in Table \ref{tbl: Res Est No QLSA} and described in more detail below. 

The total number of terms can also be exactly calculated. From \eqref{eqn:Le} and \eqref{eqn:L2aL2b}, $L^{\text{(e)}}$ is split into the $L_1^{\text{(e)}}$, $L_{2a}^{\text{(e)}}$ and $L_{2b}^{\text{(e)}}$ terms. Conveniently, $L_1^{\text{(e)}}$ is decomposed into exactly $\log n_t+1$ terms as shown in \eqref{eqn:L1e}. By combining \labelcref{eqn:L2a,eqn:Aejj,eqn:Ajj,eqn:A11} we can see that $L_{2a}^{\text{(e)}}$ has $(4\log n_x+6)\sum_{j=1}^\alpha j= \alpha(\alpha+1)(2\log n_x+3)$ terms. Similarly, by combining \labelcref{eqn:L2b,eqn:Aejp1j,eqn:A21} we can see that $L_{2b}^{\text{(e)}}$ has $4\sum_{j=1}^{\alpha-1} j= 2\alpha(\alpha-1)$ terms. By adding all three contributions together, the total number of terms in the decomposition of $L^{\text{(e)}}$ is exactly $\log n_t + 2\alpha(\alpha+1)\log n_x + \alpha(5\alpha+1) + 1$, which we say has complexity $\mathcal{O}(\log n_t + \alpha^2\log n_x)$ number of terms assuming that $n_x = n_t$ and $2 \le \alpha \ll n_x$.

Finally, to obtain an estimate of the Clifford and T gate counts, we require the following facts and assumptions (referred to as items):
\begin{enumerate}
	\item A $C^qX$ gate can be implemented with $\mathcal{O}(q)$ Clifford and $\mathcal{O}(q)$ T-gates using one dirty ancilla \cite{MultiControl}. \label{assum: CqX}
	\item An adjacent Toffoli gate is implemented with a constant number of Clifford (CNOT, H, Pauli or S) and T-gate \cite{Cruz2024}. \label{assum:Tof}
	\item An adjacent SWAP gate is implemented with three CNOT gates \cite{Hardy2006}. \label{assum:SWAP}
	\item Assume that $n_t=n_x$. 
\end{enumerate}

Here, we estimate the Clifford and T gate count for an arbitrary term in our decomposition $L_l$, by estimating the resources required for its block encoding $U_l=U_{l,1}U_{l,2}$ as described in Section \ref{BlockEncoding}. For the remainder of this section we will drop the $l$ subscript and refer to these terms as $L$, $U$, $U_1$ and $U_2$. Since $U_1$ can be implemented with a single $C^qX$ gate, as described in Section \ref{BlockEncoding}, its gate complexity scales like $\mathcal{O}(q)$ from item \ref{assum: CqX}. The procedure to determine the value $q$ is simply to count the number of operators from the set $\mathbb{P} \setminus \rho_4$ that appear in the term $LL^T$ (see \cite{GS24,GS2025_LCofThings} for a detailed explanation). Similarly, the procedure to determine $U_2$ is to obtain $\overline{L}$ and to use the relation $U_2 = I \otimes \overline{L}$. Bringing these ideas together, as an example consider the term $L = \rho_0 \otimes \rho_1 \otimes \rho_2 \otimes \rho_3 \otimes \rho_4 \otimes P$ for an arbitrary permutation matrix $P$. From this we find that $LL^T = \rho_0 \otimes \rho_0 \otimes \rho_3 \otimes \rho_3 \otimes \rho_4 \otimes \rho_4$ and $\overline{L} = \rho_4 \otimes \sigma_0 \otimes \sigma_0 \otimes \rho_4 \otimes P$. Therefore, the associated $U_1$ matrix requires one $C^4X$ gate and the associated $U_2$ circuit requires two NOT gates and the circuit for $P$.

This procedure, in conjunction with the items listed above, is applied to estimate the Clifford and T gate counts for each of the $L_1^{\text{(e)}}$, $L_{2a}^{\text{(e)}}$ and $L_{2b}^{\text{(e)}}$ terms as summarized in Table \ref{tbl: Res Est No QLSA}. To provide a concrete example about how these estimates were formed, we consider the most expensive $L_{2b}^{\text{(e)}}$ term. This term comes about by evaluating \labelcref{eqn:Aejp1j,eqn:A21} into \eqref{eqn:L2b} and choosing $j=\alpha-1$, $l=\alpha-2$, the $\rho_0^{\otimes \log n_t}$ term from \eqref{eqn:L2b}, and the $F^+$ component from \eqref{eqn:A21}. The resulting $L$, $\overline{L}$ and $LL^T$ terms are provided in the box of Figure \ref{fig:L2b Circs}. The $U_1$ term requires a single $C^qX$ gate where $q = \log \alpha n_t n_x$ since $LL^T$ has as many operators from the set $\mathbb{P} \setminus \rho_4$. Next, by looking at the $\overline{L}$ term it follows that the $U_2$ circuit requires a single NOT gate plus the circuits for $K^{(n_x^{\alpha-2},n_x)}$, $P_2^+$, $P_1$ and $K^{(n_x^2,n_x^{\alpha-2})}$. By invoking items \cref{assum: CqX,assum:Tof,assum:SWAP} and from \labelcref{eqn:P1,eqn:P2p,eqn:com_mat}, these circuits require $\mathcal{O}\left( \alpha(\log n_x)^2 \right)$ Clifford gates and $\mathcal{O}\left( (\log n_x)^2 \right)$ T gates, as appear in Table \ref{tbl: Res Est No QLSA}.


\section{Numerical Simulations} \label{NumSims}

\begin{figure}
	\centering
	\begin{subfigure}{0.48\textwidth}
		\centering
		\includegraphics[]{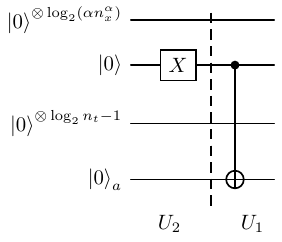}
		\vspace{-1em}
		\begin{equation*}
			\tcbhighmath{
				\begin{split}
					L &= I_{n_t/2} \otimes \rho_2 \otimes I_{\alpha n_x^\alpha} \\
					\overline{L} &= I_{n_t/2} \otimes \sigma_x \otimes I_{\alpha n_x^\alpha}\\
					L L^T &= I_{n_t/2} \otimes \rho_3 \otimes I_{\alpha n_x^\alpha} 
			\end{split}	}
		\end{equation*}
	\end{subfigure}
	
	\vspace{0.5cm}
	
	\begin{subfigure}{0.48\textwidth}
		\centering
		\includegraphics[]{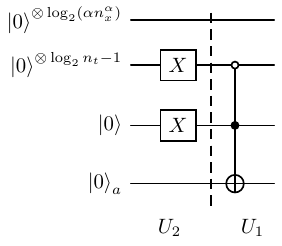}
		\vspace{-1em}
		\begin{equation*}
			\tcbhighmath{
				\begin{split}
					L &= \rho_2 \otimes \rho_1^{\otimes \log n_t -1} \otimes I_{\alpha n_x^\alpha} \\
					\overline{L} &=  \sigma_x^{\otimes \log n_t} \otimes I_{\alpha n_x^\alpha} \\
					L L^T &= \rho_3 \otimes \rho_0^{\otimes \log n_t -1} \otimes I_{\alpha n_x^\alpha}
			\end{split}	}
		\end{equation*}
	\end{subfigure}
	\caption{Circuits for the two nontrivial block encoded $L_1^{\text{(e)}}$ terms from \eqref{eqn:L1e} are given. The $\ket{0}_a$ wire is a single ancillary qubit required for the block encoding described in Section \ref{BlockEncoding}. A control operation or a single qubit gate on a multi-qubit register should be interpreted as the respective operation being applied to each individual wire within that register. The vertical dashed line separates the $U_1$ from the $U_2$ component corresponding to the block encoding discussed in Section \ref{BlockEncoding}. The equations below the circuits are: the term to block encode ($L$), its completion ($\overline{L}$) which is necessary to construct $U_2=I \otimes \overline{L}$, and $LL^T$ which is used to construct $U_1$ using the approach outlined in \cite{GS24}. The bottom circuit is formed by choosing $j=2$ in the summation. Circuits drawn using \cite{Kay2018}.}
	\label{fig:L1e Circs}
\end{figure}

In this section we implement our decomposition strategy to (i) generate circuits that load the Carleman linearized 1D Burgers' equation, and (ii) use these circuits to obtain a solution to the linear system using the VQLS method. We start with the nonlinear PDE provided in \eqref{eqn:discreteBurgers} with the following parameter settings: $n_t=4$, $n_x=4$, $\Delta t=0.25 \text{ s}$, $\Delta x = 2\pi /(n_x-1) \text{ m}$, $\nu=1.0 \text{ m}^2 \text{s}^{-1}$, and a Gaussian distribution for the initial condition given by $u(x,0) = (\sqrt{2\pi} \sigma^2)^{-1} e^{-(x - \mu)^2/(2\sigma^2)}$ with $\sigma=0.5$ and $\mu=\pi$. The Carleman linearization is then applied with a truncation order of $\alpha=2$ resulting in an equation of the form \eqref{eqn:LeYeBe} with $L^{\text{(e)}} \in \mathbb{C}^{128 \times 128}$. 

Next, we apply the decomposition strategy described in Section \ref{MatrixDecomposition}. Here, we find that $L^{\text{(e)}}$ may be decomposed into a linear combination of $73$ non-unitary matrices that are then individually block encoded using the strategy described in Section \ref{BlockEncoding}. For a point of reference, a decomposition of the same matrix using a linear combination of tensor products of Pauli matrices requires $1,142$ terms. Using our approach, there are three types of circuits stemming from the terms $L^{\text{(e)}}_1$, $L^{\text{(e)}}_{2a}$ and $L^{\text{(e)}}_{2b}$, which are ~\labelcref{eqn:L1e,eqn:L2a,eqn:L2b} respectively. A characteristic example for each of these three types of circuits are given in Figures \labelcref{fig:L1e Circs,fig:L2a Circs,fig:L2b Circs}, respectively, demonstrating that our method provides complete constructions. Circuits in the forms of Figures \labelcref{fig:L1e Circs,fig:L2a Circs,fig:L2b Circs} are used to generate all of the $73$ aforementioned terms in the linear combination of $L^{\text{(e)}}$. 

\begin{figure}
	\centering

	\includegraphics[]{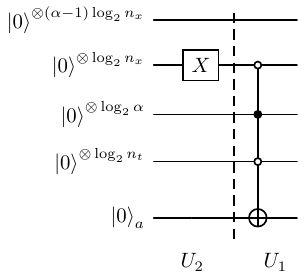}
	\vspace{-1em}
	\begin{equation*}
		\tcbhighmath{
			\begin{split}
				L &= \rho_0^{\otimes \log n_t} \otimes \rho_3^{\otimes \log \alpha} \otimes \rho_1^{\otimes \log n_x} \otimes I_{n_x^{\alpha-1}} \\
				\overline{L} &= I_{\alpha n_t} \otimes \sigma_x^{\otimes \log n_x} \otimes I_{n_x^{\alpha-1}} \\
				L L^T &= \rho_0^{\otimes \log n_t} \otimes \rho_3^{\otimes \log \alpha} \otimes \rho_0^{\otimes \log n_x} \otimes I_{n_x^{\alpha-1}}  
		\end{split}	}
	\end{equation*}
	
	\caption{Same as Figure \ref{fig:L1e Circs}, except for the $L^{\text{(e)}}_{2a}$ term where we have evaluated both \labelcref{eqn:Ajj,eqn:A11} into \eqref{eqn:L2a} and chosen $j=\alpha$, $l=0$, the $\rho_0^{\otimes \log n_t}$ term from \eqref{eqn:L2a}, and the $\rho_1^{\otimes s}$ term ($s=\log n_x$) from \eqref{eqn:A11}. The circuit will take a similar form for other values of $j$, $l$ and if other terms from \eqref{eqn:A11} are used. Note that the circuit has the same width as those from Figure \ref{fig:L1e Circs}.}
	\label{fig:L2a Circs}
\end{figure}

\begin{figure*}
	\centering
	\includegraphics[]{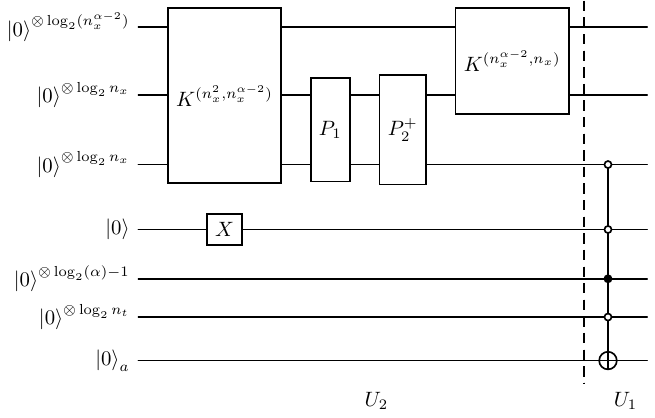}
	\vspace{-1em}
	\begin{equation*}
	\tcbhighmath{
		\begin{split}
			L &= \rho_0^{\otimes \log n_t} \otimes \rho_3^{\otimes \log(\alpha)-1} \otimes \rho_1 
			\otimes \\
			&\qquad 
			\left[ \left( \rho_0^{\otimes \log n_x} \otimes K^{(n_x^{\alpha-2},n_x)} \right) 
			\cdot \left(
			\left( \rho_0^{\otimes \log n_x} \otimes I_{n_x} \right)
			P_2^+ P_1 \otimes I_{n_x^{\alpha-2}} \right) 
			\cdot K^{(n_x^2,n_x^{\alpha-2})} \right] \\[5pt]
			\overline{L} &=  I_{\alpha n_t/2} \otimes \sigma_x 
			\otimes \left[ \left( I_{n_x} \otimes K^{(n_x^{\alpha-2},n_x)} \right) 
			\cdot \left( P_2^+ P_1 \otimes I_{n_x^{\alpha-2}} \right) 
			\cdot K^{(n_x^2,n_x^{\alpha-2})} \right] \\[5pt]
			L L^T &= \rho_0^{\otimes \log n_t} \otimes \rho_3^{\otimes \log(\alpha)-1} 
			\otimes \rho_0^{\otimes \log(n_x)+1} \otimes I_{n_x^{\alpha-1}}
		\end{split}	}
	\end{equation*}
	\caption{Same as Figure \ref{fig:L1e Circs}, except for the $L^{\text{(e)}}_{2b}$ term where we have evaluated  \labelcref{eqn:Aejp1j,eqn:A21} into \eqref{eqn:L2b} and chosen $j=\alpha-1$, $l=\alpha-2$, the $\rho_0^{\otimes \log n_t}$ term from \eqref{eqn:L2b}, and the $F^+$ component from \eqref{eqn:A21}. The circuit will take a similar form for other values of $j$, $l$ and if the $F^-$ term is used instead of $F^+$. The circuits for $P_1$ and $P_2^+$ are given in \labelcref{eqn:P1,eqn:P2p} respectively, and the circuits for the commutation matrices $K^{(a,b)}$ are given in \eqref{eqn:com_mat}. Note that the circuit has the same width as those from Figures \labelcref{fig:L1e Circs,fig:L2a Circs}.}
	\label{fig:L2b Circs}
\end{figure*}

Next, we combine our approach with the VQLS algorithm from \cite{VQLS} to obtain a solution to the linear system in \eqref{eqn:LeYeBe} for this specific setup. The VQLS method is a variational approach whereby an ansatz ($V(\vec{\theta})$) with variational parameters ($\vec{\theta}$) is used to approximate the true solution ($\vec{Y}^{\text{(e)}}$) by searching the parameter space to find the optimal parameters ($\vec{\theta}_\text{opt}$) such that $V(\vec{\theta}_\text{opt}) \ket{0} \approx \vec{Y}^{\text{(e)}}$. To do this, the quantum computer is tasked with calculating the classically intractable expectation values that are then passed to a classical computer to (i) compute a cost function, and (ii) update the variational parameters using an optimization routine. Given the updated variational parameters, the quantum computer recomputes the expectation values and again passes the information back to the classical computer to repeat (i) and (ii). This back-and-forth process is repeated until the cost function reaches a predefined minimum at which point $V(\vec{\theta}_\text{opt})$ is obtained. 

To implement the VQLS routine one requires a cost function, ansatz, and an optimization routine. First, to avoid the barren plateau phenomena \cite{Larocca2025}, we have elected to use the local cost function outlined in \cite{VQLS}, which, if we consider the decomposition from \eqref{eqn:LCNU}, is written as
\begin{equation*}
	\begin{split}
		C(\vec{\theta}) &= \frac{1}{2} \left(1 - \frac{1}{n} 
		\frac{\sum_{k=0}^{n-1}\sum_{l,l^\prime=0}^{N_s-1} c_l c_{l^\prime}^* \delta_{l,l^\prime}^k}
		{\sum_{l,l^\prime=0}^{N_s-1} c_l c_{l^\prime}^* \beta_{l,l^\prime}} \right) \\
		\delta_{l,l^\prime}^k &= \bra{V(\vec{\theta})} L_{l^\prime}^\dagger U_b Z_k U_b^\dagger L_l \ket{V(\vec{\theta})} \\
		\beta_{l,l^\prime} &= \bra{V(\vec{\theta})} L_{l^\prime}^\dagger L_l \ket{V(\vec{\theta})} \,,
	\end{split}
\end{equation*}
where $n$ is the number of qubits, $N_s=73$ for the parameters selected in this section, the coefficients $c_l$ come from \eqref{eqn:LCNU}, $U_b\ket{0}=\vec{B}^{\text{(e)}}$ loads the RHS of \eqref{eqn:LeYeBe}, and $Z_k$ is the Pauli-Z matrix applied on the $k^\text{th}$ qubit. Next, for our ansatz we use Circuit 18 from \cite{Sim2019} and reproduced in Figure \ref{fig:Ansatz_Sim18}, which is hardware efficient with decent expressibility and entangling capabilities. Since the dimension of this test problem is small, we need only three ansatz layers yielding $3 n=21$ total variational parameters for number of qubits $n=\log(\alpha n_t n_x^\alpha)$. Finally, we use the conjugate gradient optimizer available in IBM \texttt{Qiskit's} software package \texttt{qiskit\_algorithms.optimizers} with a tolerance setting of $10^{-3}$. The conjugate gradient optimizer was used because it was found to perform the best among twelve commonly used optimizers for idealized (noiseless) settings \cite{Singh2023} such as the one used here. It is important to note that there are better choices for noisy settings. 

\begin{figure*}
	\centering
	\includegraphics{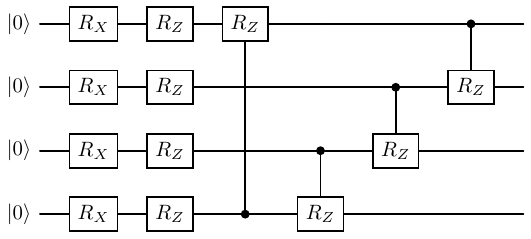}
	\caption{A single layer of the ansatz used for the VQLS routine taken from Circuit 18 from \cite{Sim2019}.} 
	\label{fig:Ansatz_Sim18}
\end{figure*}

Bringing everything together, the VQLS solution to the Carleman linearized 1D Burgers' equation using our novel decomposition strategy is shown in Figure \ref{fig:Burgers_Solution}. For validation, we compare our quantum solution to a classical solution whereby we have solved exactly the same linear system, but have instead used the \texttt{python} function \texttt{numpy.linalg.solve} to obtain a solution. A visual comparison shows that the quantum solution validates well to the classical solution, thereby demonstrating that our decomposition strategy is capable of producing accurate results. 

\begin{figure*}
	\centering
	\includegraphics[scale=0.5]{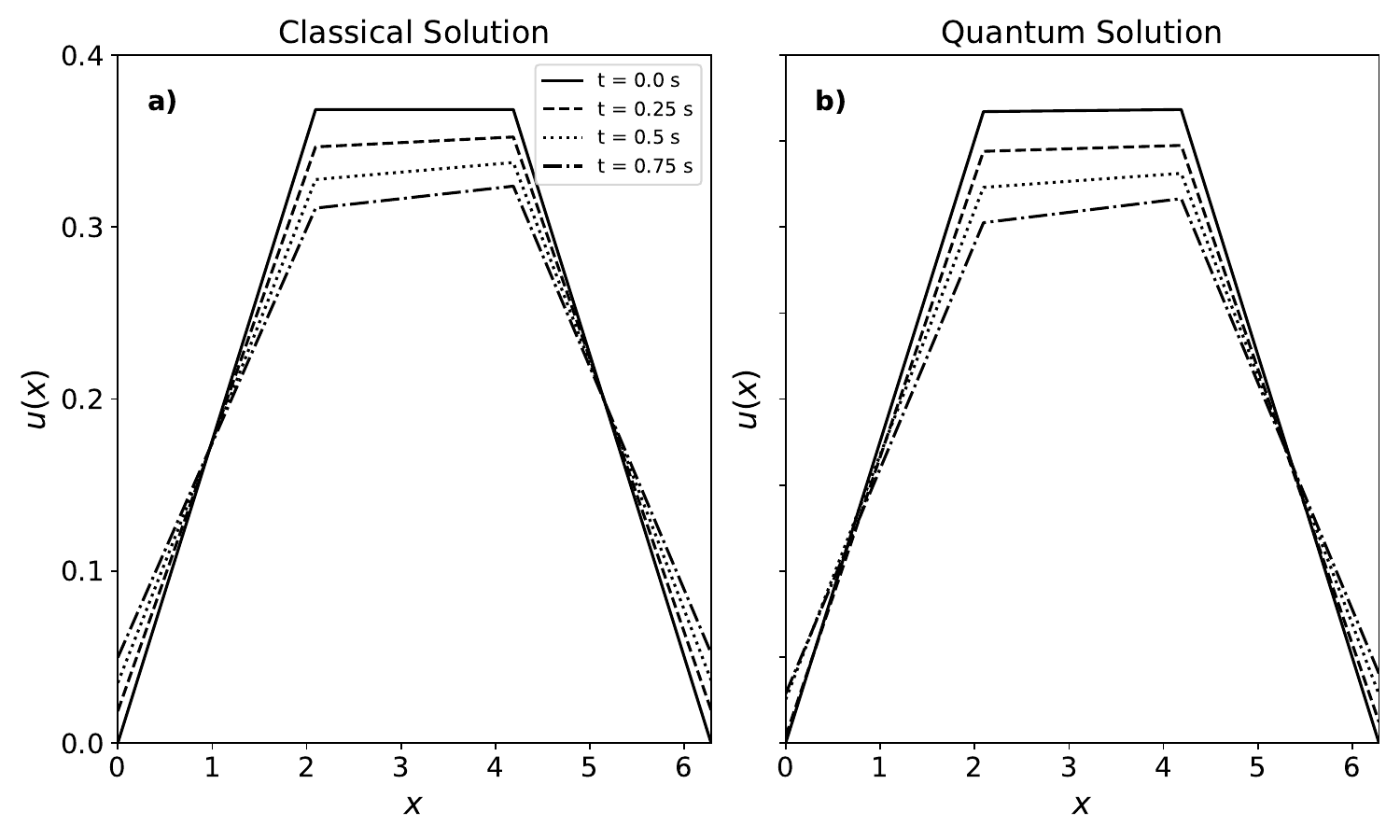}
	\caption{Solutions to the Carleman linearized 1D Burgers equation from \eqref{eqn:LeYeBe} using (a) a classical linear systems solver, and (b) our decomposition approach combined with VQLS \cite{VQLS}. The initial condition is the solid black line and the evolution at integer multiples of $0.25 \text{ s}$ is provided in the dashed and dotted lines.} 
	\label{fig:Burgers_Solution}
\end{figure*}


\section{Discussion and Conclusions} \label{Conclusions}
In this work, we solve the decomposition problem for the 1D Carleman linearized Burgers' equation. The key insights introduced in this study are two-fold: (1) to embed the original Carleman system into an even larger system of equations, and (2) to extend the methods introduced in \cite{GS24} to include products of elements from $\mathbb{P}$ with specific unitary matrices. The advantage gained from these insights is that the larger system can be decomposed into a linear combination of $\mathcal{O}(\log n_t + \alpha^2\log n_x)$ terms, whereas the original has no known polylogarithmic decomposition. While these terms are non-unitary, they can be efficiently block encoded into unitary matrices, and therefore used in a QLSA. As an example, we use the VQLS where we found that the upper bounds for the Clifford and T gate counts are $\mathcal{O}(\alpha(\log n_x)^2)$ and $\mathcal{O}((\log n_x)^2)$ respectively. Together, these polylogarithmic scalings suggest that it may be possible to exponentially increase the spatial and temporal grid sizes in CFD and NWP models. That being said, whether an exponential increase is possible is still an open question and there are still major challenges that must be solved. 

One such challenge with the Carleman linearized Burgers' (or Navier-Stokes) equation is that, for strongly nonlinear interactions, it may not be possible to efficiently find accurate solutions, as put forth by \cite{Liu21}. However, this may be a case of learning through experiment since their empirical results do differ from their analytical results. One way to completely circumvent the strong nonlinearity issue is to apply the Carleman linearization method to the Lattice-Boltzman equation (LBE) rather than the Navier-Stokes \cite{li2025potential}. The advantage being that the LBE is inherently weakly nonlinear provided that the Mach number is small. It is therefore important to note that while the present study focused on the 1D Burgers' equation, the zero padding method introduced here can also be applied to the LBE. In fact, a related embedding technique was introduced in the encoding oracles in \cite{penuel2024feasibility}. The work presented here therefore fits well into the quantum algorithm literature by contributing a generalizable method useful for different approaches.

Another such challenge is related to the data readout problem, whereby quantum advantage may be lost when extracting exponentially small probability amplitudes from the state vector. In Liu \textit{et al.}\! \cite{Liu21}, they find that for dissipative systems ``we must have a suitable initial condition and terminal time such that the final state is not exponentially smaller than the initial state." The reasoning comes about because if the final solution is exponentially smaller than the initial state, then reading the probability amplitudes of the final state could take an exponential number of final state preparations. If instead we consider a dissipative system with a forcing term, then the final state is not necessarily exponentially smaller than the initial state, and therefore readout can be efficient. Given that, quantum advantage for the homogeneous case may not be possible if the terminal time $t_\text{final}$ is too large such that the amplitude of the solution at $t_\text{final}$ is exponentially smaller than that of the initial condition. Therefore, $t_\text{final}$ will depend upon both the initial state and the diffusion coefficient $\nu$, and must be chosen carefully to ensure that the final solution is not exponentially smaller than the initial state.

Finally, an important limitation of this work is that we make no effort to transpile our circuits since that is a device specific process. While our decompositions do achieve the desirable polylogarithmic circuit depth complexity, we implicitly assume an all-to-all connectivity for the topology. A different topology will introduce more overhead and, since the overhead is device specific, we cannot make general remarks on how this will impact our circuit depth complexities. That being said, the circuit depths reported here are a starting point and are certainly not optimal. In fact, there are known improvements to at least two of the circuits used, that are the incrementer \cite{Incrementer} and the decomposition of the $C^jX$ gates \cite{MultiControl}. So, while the transpilation of our circuits onto real hardware will incur some overhead, there is also reason to believe that we can even potentially improve upon the circuit depths reported here.


\section*{Acknowledgements}
We gratefully acknowledge the support NRL Base Program PE 0601153N. Furthermore, we are thankful for the stimulating discussions with Dr. Arzhang Angoshtari.


\onecolumngrid
\appendix


\section{Example of the Quaternary Mappings} \label{QuaternaryMapping}
Table \ref{table:QuatMapTable} lists some arbitrary examples of the quaternary mapping method introduced in Section \ref{CarlemanEmbedding}. The purpose of these terms is to place the $A_{j}^{\text{(e)},j}$ and $A_{j+1}^{\text{(e)},j}$ terms in \eqref{eqn:Ae} in their appropriate positions along the diagonal and super-diagonal respectively.

\begin{table}[h]
\begin{center}
\renewcommand{\arraystretch}{1.5}
\begin{tabular}{m{0.025\linewidth}| m{0.025\linewidth} | m{0.025\linewidth} | m{0.05\linewidth} | m{0.05\linewidth} | m{0.15\linewidth} | m{0.15\linewidth}}
 \hline
 $\alpha$ & $i$ & $j$ & $b_\alpha(i)$ & $b_\alpha(j)$ & $f(b_\alpha(i),b_\alpha(j))$ & $\rho_{f_{K-1}\dots f_0}$ \\
 \hline \hline
 $2$ & $0$ & $0$ & $0$ & $0$ & $0$ & $\rho_0$ \\
 \hline
 $2$ & $0$ & $1$ & $0$ & $1$ & $1$ & $\rho_1$ \\
 \hline 
 $4$ & $0$ & $1$ & $00$ & $01$ & $01$ & $\rho_0\otimes\rho_1$ \\
 \hline
 $4$ & $2$ & $3$ & $10$ & $11$ & $31$ & $\rho_3\otimes\rho_1$ \\
 \hline
 $8$ & $1$ & $5$ & $001$ & $101$ & $103$ & $\rho_1\otimes\rho_0\otimes\rho_3$ \\
 \hline  
 $8$ & $6$ & $7$ & $110$ & $111$ & $331$ & $\rho_3\otimes\rho_3\otimes\rho_1$ \\
 \hline
\end{tabular}
\end{center}
\caption{Some arbitrary examples for the quaternary mapping method discussed in Section \ref{CarlemanEmbedding}. Here, $\alpha$ is the truncation order, $i,j$ are matrix element indices, $b_\alpha(k)$ is the decimal to binary mapping function of bitstring length $\log\alpha$, $f(b_\alpha(i),b_\alpha(j))$ maps the binary values to their respective quaternary values, and $\rho_{f_{K-1}\dots f_0}$ are the full products using the quaternary bitstrings.}
\label{table:QuatMapTable}
\end{table}


\section{Derivation for $A_{j+1}^{\textnormal{(e)},j}$ \eqref{eqn:Aejp1j}} \label{DerivationAejp1j}
Here, we derive the $A_{j+1}^{\text{(e)},j}$ equation for $j=\{1,\dots,\alpha-1\}$. First, we expand \eqref{eqn:Ajp1j} using the property $A\otimes B = K^{(r,m)}\cdot(B\otimes A)\cdot K^{(n,q)}$ where $A\in\mathbb{C}^{r\times q}$, $B\in\mathbb{C}^{m\times n}$,  and $K^{(a,b)}\in\mathbb{C}^{ab\times ab}$ is the commutation matrix \cite{wiki:Commutation_matrix, Watrous2018}. Using the definition of $A_{j+1}^j$ from \eqref{eqn:Ajp1j}, this gives 
\begin{equation} \label{eqn:Ajp1j_com}
\begin{split}
    A_{j+1}^j &= \sum_{l=0}^{j-1} I_{n_x}^{\otimes l} \otimes F_2 \otimes I_{n_x}^{\otimes j-l-1} \\
    &= \sum_{l=0}^{j-1} \Bigl( K^{(n_x^l,n_x)} \cdot (F_2 \otimes I_{n_x}^{\otimes l}) \cdot K^{(n_x^2,n_x^l)} \Bigr) 
    \otimes I_{n_x}^{\otimes j-l-1} \,.
\end{split}
\end{equation}
Next, we evaluate \eqref{eqn:Ajp1j_com} into \eqref{eqn:Aejp1j} to obtain
\begin{equation} \label{eqn:Aejp1jDerivation}
\begin{split}
    A_{j+1}^{\text{(e)},j} &\coloneq 
    \begin{pmatrix}
        A_{j+1}^j & 0_{n_x^j \times (n_x^\alpha-n_x^{j+1})} \\
        0_{(n_x^\alpha-n_x^j)\times n_x^{j+1}} & 0_{(n_x^\alpha-n_x^j)\times(n_x^{\alpha}-n_x^{j+1})}
    \end{pmatrix} \\[4pt]
    &=
    \begin{pmatrix}
        \sum_{l=0}^{j-1} \Bigl( K^{(n_x^l,n_x)} \cdot (F_2 \otimes I_{n_x}^{\otimes l}) \cdot K^{(n_x^2,n_x^l)} \Bigr) \otimes I_{n_x}^{\otimes j-l-1} 
        & 0_{n_x^j \times (n_x^\alpha-n_x^{j+1})} \\
        0_{(n_x^\alpha-n_x^j)\times n_x^{j+1}} & 0_{(n_x^\alpha-n_x^j)\times(n_x^{\alpha}-n_x^{j+1})}
    \end{pmatrix} \\[4pt]
    &=
    \sum_{l=0}^{j-1}
    \begin{pmatrix}
        \Bigl( K^{(n_x^l,n_x)} \cdot (F_2 \otimes I_{n_x}^{\otimes l}) \cdot K^{(n_x^2,n_x^l)} \Bigr) \otimes I_{n_x}^{\otimes j-l-1} 
        & 0_{n_x^j \times (n_x^\alpha-n_x^{j+1})} \\
        0_{(n_x^\alpha-n_x^j)\times n_x^{j+1}} & 0_{(n_x^\alpha-n_x^j)\times(n_x^{\alpha}-n_x^{j+1})}
    \end{pmatrix} \\[4pt]
    &=
    \sum_{l=0}^{j-1}
    \begin{pmatrix}
        \Bigl( K^{(n_x^l,n_x)} \cdot (F_2 \otimes I_{n_x}^{\otimes l}) \cdot K^{(n_x^2,n_x^l)} \Bigr) 
        & 0_{n_x^{l+1}\times(n_x^{\alpha-j+l+1}-n_x^{l+2})} \\
        0_{(n_x^{\alpha-j+l+1}-n_x^{l+1})\times n_x^{l+2}} & 0_{(n_x^{\alpha-j+l+1}-n_x^{l+1})\times (n_x^{\alpha-j+l+1}-n_x^{l+2})}
    \end{pmatrix} 
    \otimes I_{n_x}^{\otimes j-l-1} \\[4pt]
    &=
    \rho_0^{\otimes\log(n_x^{\alpha-j-1})} \otimes
    \sum_{l=0}^{j-1}
    \begin{pmatrix}
        K^{(n_x^l,n_x)} \cdot (F_2 \otimes I_{n_x}^{\otimes l}) \cdot K^{(n_x^2,n_x^l)} \\
        0_{(n_x^{l+2}-n_x^{l+1})\times n_x^{l+2}}
    \end{pmatrix} 
    \otimes I_{n_x}^{\otimes j-l-1} \,.
\end{split}
\end{equation}

\noindent Next, we simplify the matrix product terms by
\begin{equation} \label{eqn:MatProd}
\begin{split}
    &\begin{pmatrix}
        K^{(n_x^l,n_x)} \cdot (F_2 \otimes I_{n_x}^{\otimes l}) \cdot K^{(n_x^2,n_x^l)} \\
        0_{(n_x^{l+2}-n_x^{l+1})\times n_x^{l+2}}
    \end{pmatrix}
    =
    \begin{pmatrix}
        K^{(n_x^l,n_x)} \cdot (F_2 \otimes I_{n_x}^{\otimes l}) \\
        0_{(n_x^{l+2}-n_x^{l+1})\times n_x^{l+2}}
    \end{pmatrix} 
    \cdot K^{(n_x^2,n_x^l)} \\[4pt]
    &\quad=
    \begin{pmatrix}
        K^{(n_x^l,n_x)} & 0_{n_x^{l+1}\times(n_x^{l+2}-n_x^{l+1})} \\
        0_{(n_x^{l+2}-n_x^{l+1})\times n_x^{l+1}} & 0_{(n_x^{l+2}-n_x^{l+1})\times (n_x^{l+2}-n_x^{l+1})}
    \end{pmatrix}
    \cdot 
     \begin{pmatrix}
        F_2 \otimes I_{n_x}^{\otimes l} \\
        0_{(n_x^{l+2}-n_x^{l+1})\times n_x^{l+2}}
    \end{pmatrix} 
    \cdot K^{(n_x^2,n_x^l)} \\[4pt]
    &\quad=
    \Bigl( \rho_0^{\otimes \log n_x} \otimes K^{(n_x^l,n_x)} \Bigr) \cdot 
    \Biggl(
     \begin{pmatrix}
        F_2 \\
        0_{(n_x^2-n_x)\times n_x^2}
    \end{pmatrix} 
    \otimes I_{n_x}^{\otimes l} 
    \Biggr)    
    \cdot K^{(n_x^2,n_x^l)} \,.
\end{split}
\end{equation}

\noindent Finally, evaluate \eqref{eqn:MatProd} into \eqref{eqn:Aejp1jDerivation} to give the full expression
\begin{equation} \notag
\begin{split}
    A_{j+1}^{\text{(e)},j} &= \rho_0^{\otimes\log(n_x^{\alpha-j-1})} \\
    &\qquad\otimes
    \sum_{l=0}^{j-1}
    \Biggl[
    \Bigl( \rho_0^{\otimes \log n_x} \otimes K^{(n_x^l,n_x)} \Bigr) \cdot 
    \Biggl(
     \begin{pmatrix}
        F_2 \\
        0_{(n_x^2-n_x)\times n_x^2}
    \end{pmatrix} 
    \otimes I_{n_x}^{\otimes l} 
    \Biggr)
    \cdot K^{(n_x^2,n_x^l)} \Biggr]
    \otimes I_{n_x}^{\otimes j-l-1} \,.
\end{split}
\end{equation}


\section{Example Case: $n_x=4$} \label{Example_n4}
For $n_x=4$ we have 
\begin{subequations}
\label{equations}
\begin{align}
    \label{eqn:A21p}
    F_2^+ &=
    \begin{pmatrix}
        0 & 1 & 0 & 0  & 0  & 0 & 0 & 0 & 0 & 0  & 0 & 0 & 0 & 0 & 0  & 0 \\
        0 & 0 & 0 & 0  & 0  & 0 & 1 & 0 & 0 & 0  & 0 & 0 & 0 & 0 & 0  & 0 \\
        0 & 0 & 0 & 0  & 0  & 0 & 0 & 0 & 0 & 0  & 0 & 1 & 0 & 0 & 0  & 0 \\
        0 & 0 & 0 & 0  & 0  & 0 & 0 & 0 & 0 & 0  & 0 & 0 & 1 & 0 & 0  & 0 \\    
    \end{pmatrix}, \\[4pt]
    \label{eqn:A21m}
    F_2^- &=
    \begin{pmatrix}
        0 & 0 & 0 & 1 & 0  & 0 & 0 & 0 & 0 & 0  & 0 & 0 & 0 & 0 & 0  & 0 \\
        0 & 0 & 0 & 0  & 1 & 0 & 0 & 0 & 0 & 0  & 0 & 0 & 0 & 0 & 0  & 0 \\
        0 & 0 & 0 & 0  & 0  & 0 & 0 & 0 & 0 & 1 & 0 & 0 & 0 & 0 & 0  & 0 \\
        0 & 0 & 0 & 0  & 0  & 0 & 0 & 0 & 0 & 0  & 0 & 0 & 0 & 0 & 1 & 0 \\    
    \end{pmatrix} \,,
\end{align}
\end{subequations}
where $F_2=-(F_2^+-F_2^-)/(2\Delta x)$. Therefore, the full matrix is given by

\begin{equation} \notag
    \begin{pmatrix} F_2 \\ 0_{(n_x^2-n_x) \times n_x^2} \end{pmatrix}
    = \frac{-1}{2\Delta x}
    \begin{pmatrix}
        0 & 1 & 0 & -1 & 0  & 0 & 0 & 0 & 0 & 0  & 0 & 0 & 0 & 0 & 0  & 0 \\
        0 & 0 & 0 & 0  & -1 & 0 & 1 & 0 & 0 & 0  & 0 & 0 & 0 & 0 & 0  & 0 \\
        0 & 0 & 0 & 0  & 0  & 0 & 0 & 0 & 0 & -1 & 0 & 1 & 0 & 0 & 0  & 0 \\
        0 & 0 & 0 & 0  & 0  & 0 & 0 & 0 & 0 & 0  & 0 & 0 & 1 & 0 & -1 & 0 \\
        \hline
        0 & 0 & 0 & 0 & 0 & 0 & 0 & 0 & 0 & 0 & 0 & 0 & 0 & 0 & 0 & 0\\
        0 & 0 & 0 & 0 & 0 & 0 & 0 & 0 & 0 & 0 & 0 & 0 & 0 & 0 & 0 & 0\\
        0 & 0 & 0 & 0 & 0 & 0 & 0 & 0 & 0 & 0 & 0 & 0 & 0 & 0 & 0 & 0\\
        0 & 0 & 0 & 0 & 0 & 0 & 0 & 0 & 0 & 0 & 0 & 0 & 0 & 0 & 0 & 0\\
        0 & 0 & 0 & 0 & 0 & 0 & 0 & 0 & 0 & 0 & 0 & 0 & 0 & 0 & 0 & 0\\
        0 & 0 & 0 & 0 & 0 & 0 & 0 & 0 & 0 & 0 & 0 & 0 & 0 & 0 & 0 & 0\\
        0 & 0 & 0 & 0 & 0 & 0 & 0 & 0 & 0 & 0 & 0 & 0 & 0 & 0 & 0 & 0\\
        0 & 0 & 0 & 0 & 0 & 0 & 0 & 0 & 0 & 0 & 0 & 0 & 0 & 0 & 0 & 0\\
        0 & 0 & 0 & 0 & 0 & 0 & 0 & 0 & 0 & 0 & 0 & 0 & 0 & 0 & 0 & 0\\
        0 & 0 & 0 & 0 & 0 & 0 & 0 & 0 & 0 & 0 & 0 & 0 & 0 & 0 & 0 & 0\\
        0 & 0 & 0 & 0 & 0 & 0 & 0 & 0 & 0 & 0 & 0 & 0 & 0 & 0 & 0 & 0\\
        0 & 0 & 0 & 0 & 0 & 0 & 0 & 0 & 0 & 0 & 0 & 0 & 0 & 0 & 0 & 0\\        
    \end{pmatrix} \,.
\end{equation}


\section{Derivation of the Permutation Matrices: $P_1, P_2^+,\text{ and } P_2^-$} \label{Permutation_Matrices}
Figure \ref{fig:P2P1} shows an example of how to create the $P^+=P_2^+P_1$ matrix for the $n_x=4$ case (a similar procedure exists for $P^-$). The $P_1$ operation is straightforward to implement using the $\times(n_x+1)\Mod{n_x^2}$ modular multiplication circuit given in \eqref{eqn:P1}. Since $n_x+1$ is odd, \cite{ModMult} provides a general implementation of the necessary modular multiplication circuit. However, we can considerably reduce the complexity of their circuit for our needs since the bottom $n_x^2-n_x$ rows are non-unique. The only limitation of these latter rows, represented by $C_2$ in Figure \ref{fig:P2P1}, is that they must form a unitary complement to the first $n_x$ rows, as discussed in Section \ref{MatrixDecomposition}.

Next, the $P_2^+$ operation from \eqref{eqn:P2p} increments each non-zero element forward by one. This is straightforward for each element except the $(n_x-1,n_x^2-1)^\text{th}$ element, which must be carried over. Conveniently, we can also simultaneously satisfy the periodic boundary condition if carried over to the $(n_x-1,n_x^2-n_x)^\text{th}$ element. This type of carryover is achieved by applying the incrementer on the first $\log n_x$-qubits. Note that while there are many different incrementer circuits as discussed in \cite{Incrementer}, we have chosen the multi-control NOT incrementer simply as a starting point to be improved upon later.

\begin{figure}[h]
  \centering
  \includegraphics[scale=0.5]{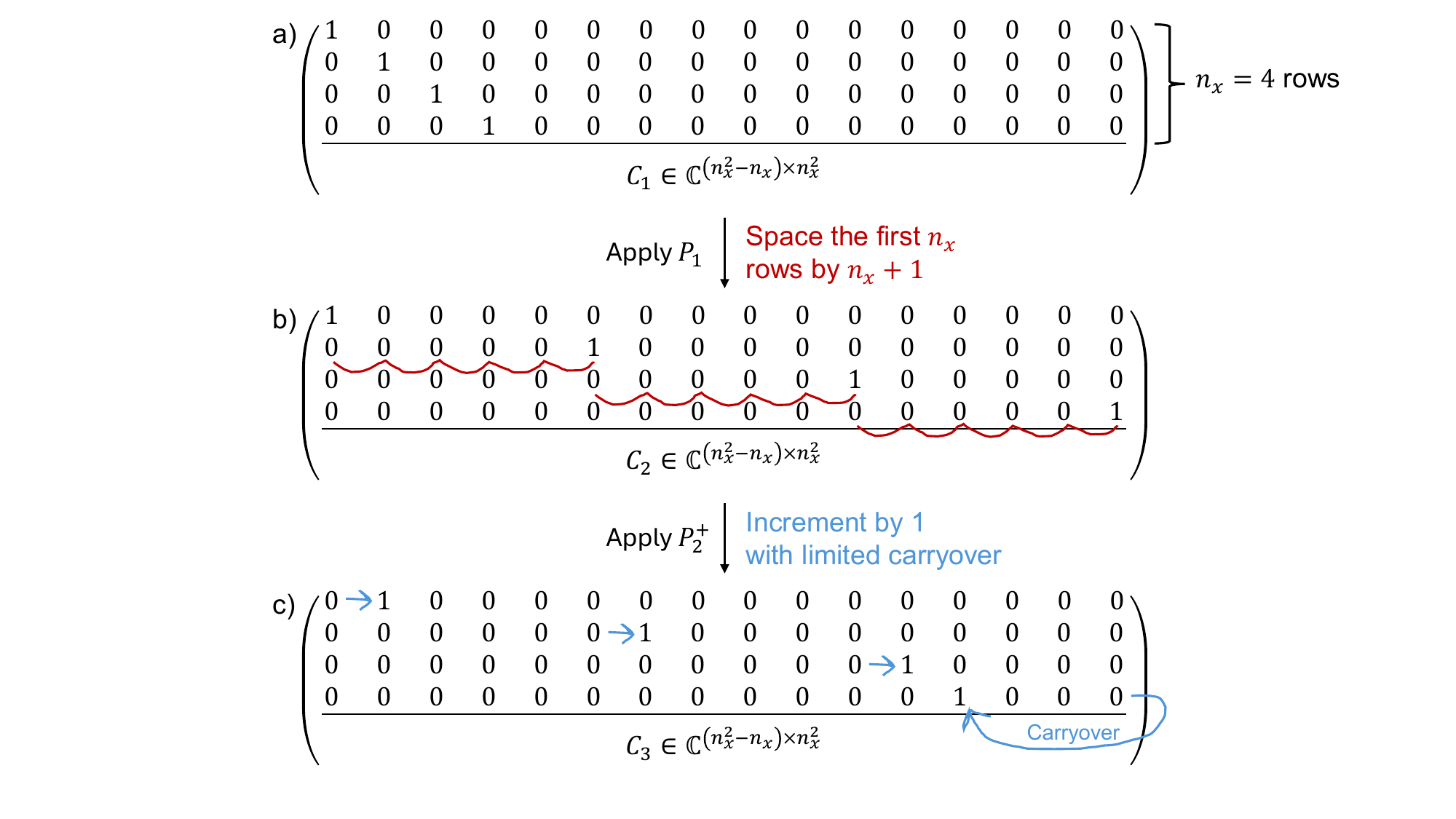}
  \caption{Operations to create the $P^+$ matrix for the $n_x=4$ case. The $P_1$ matrix performs the $\times(n_x+1)\Mod{n_x^2}$ modular multiplication to transform (a) the identity matrix into (b) an intermediary matrix where the first $n_x$ non-zero elements are spaced by $n_x+1$. Note that the $C_1$ matrix is the lower $n_x^2-n_x$ portion of the identity matrix in (a), and that the $C_2$ matrix is a non-unique unitary complement to the upper $n_x\times n_x^2$ elements in (b). Next, the $P_2^+$ matrix increments each non-zero element by one with limited carryover to transform the intermediary matrix into (c) the $P^+$ matrix. Once again, the $C_3$ matrix is a non-unique unitary complement to the upper $n_x\times n_x^2$ elements. There is an analogous transformation to prepare $P^-$.}
  \label{fig:P2P1}
\end{figure}


\section{Proof for Theorem \ref{ACompletionTheorem}} \label{ACompletionProof}
\begin{proof}
Here we prove that the unitary completion of $\mathcal{A}$, as defined in \eqref{eqn:GeneralAejp1j}, is $\overline{\mathcal{A}}$, as defined in \eqref{eqn:ABar}. As a consequence of Definition \ref{DefCompletionComplement}, if $\overline{\mathcal{A}}$ is the unitary completion to $\mathcal{A}$, then $U$ is unitary where 
\begin{equation} \notag
    U = \begin{pmatrix} \mathcal{A}^c & \mathcal{A} \\ \mathcal{A} & \mathcal{A}^c \end{pmatrix} \,.
\end{equation}
Therefore, to prove that $\overline{\mathcal{A}}$ is the unitary completion to $\mathcal{A}$, it is sufficient to show that $U$ is unitary. Since all of the matrices used in our particular decomposition are real, the conjugate transpose is equivalent to the transpose. Therefore, we start with
\begin{equation} \label{eqn:UUT}
\begin{split}
    UU^T &= (I \otimes \mathcal{A}^c + \sigma_0\otimes \mathcal{A})(I \otimes \mathcal{A}^{cT} + \sigma_0\otimes \mathcal{A}^T) \\
    &= I\otimes \mathcal{A}^c\mathcal{A}^{cT} + \sigma_0\otimes \mathcal{A}^c\mathcal{A}^T + \sigma_0\otimes \mathcal{A}\mathcal{A}^{cT} + I\otimes \mathcal{A}\mathcal{A}^T  \,.
\end{split}
\end{equation}
From Definition \ref{DefCompletionComplement}, the first term can be expanded using
\begin{equation}
    \mathcal{A}^c\mathcal{A}^{cT} = (\overline{\mathcal{A}}-\mathcal{A})(\overline{\mathcal{A}}^T-\mathcal{A}^T) \,,
\end{equation}
where
\begin{equation} \notag
\begin{split}
    \overline{\mathcal{A}}\overline{\mathcal{A}}^T &= \Bigl( \bigotimes_{k=0}^{Q_1-1} \overline{\rho}_{r_k}\overline{\rho}_{r_k}^T \Bigr) 
    \otimes I_{n_x}^{\otimes j+1} \,, \\
    \overline{\mathcal{A}}\mathcal{A}^T &= \Bigl( \bigotimes_{k=0}^{Q_1-1} \overline{\rho}_{r_k}\rho_{r_k}^T \Bigr) 
    \otimes \mathcal{D} \otimes I_{n_x}^{\otimes j-1} \,, \\
    \mathcal{A}\overline{\mathcal{A}}^T &= \Bigl( \bigotimes_{k=0}^{Q_1-1} \rho_{r_k}\overline{\rho}_{r_k}^T \Bigr) 
    \otimes \mathcal{D} \otimes I_{n_x}^{\otimes j-1} \,, \\
    \mathcal{A}\mathcal{A}^T &= \Bigl( \bigotimes_{k=0}^{Q_1-1} \rho_{r_k}\rho_{r_k}^T \Bigr) 
    \otimes \mathcal{D} \otimes I_{n_x}^{\otimes j-1} \,,
\end{split}
\end{equation}
where we have used the mixed-product property. These relations can be simplified. First, since $\overline{\rho}_{r_k}\in\{\sigma_0,\sigma_3\}$, then $\overline{\rho}_{r_k}\overline{\rho}_{r_k}^T=I$ and therefore $\overline{\mathcal{A}}\overline{\mathcal{A}}^T=I$. Next, using the result from Table \ref{table:ACompTable1}, it follows that $\overline{\mathcal{A}}\mathcal{A}^T=\mathcal{A}\overline{\mathcal{A}}^T=\mathcal{A}\mathcal{A}^T$ . Putting this all together yields $\mathcal{A}^c\mathcal{A}^{cT}=I-\mathcal{A}\mathcal{A}^T$. Using these same properties we have
\begin{equation} \notag
    \mathcal{A}^c\mathcal{A}^T = (\overline{\mathcal{A}}-\mathcal{A})\mathcal{A}^T = \overline{\mathcal{A}}\mathcal{A}^T-\mathcal{A}\mathcal{A}^T = 0 \,,
\end{equation}
and
\begin{equation} \notag
    \mathcal{A}\mathcal{A}^{cT} = \mathcal{A}(\overline{\mathcal{A}}^T - \mathcal{A}^T) = \mathcal{A}\overline{\mathcal{A}}^T - \mathcal{A}\mathcal{A}^T = 0 \,.
\end{equation}

\begin{table}
\begin{center}
\renewcommand{\arraystretch}{1.5}
\begin{tabular}{m{0.08\linewidth}| m{0.1\linewidth} | m{0.1\linewidth} | m{0.12\linewidth} | m{0.1\linewidth} }
 \hline
   & $\rho_{r_k}=\rho_0$ & $\rho_{r_k}=\rho_1$ & $\rho_{r_k}=\rho_2$ & $\rho_{r_k}=\rho_3$\\ 
 \hline\hline
 $\overline{\rho}_{r_k} \rho_{r_k}^T=$ & $\rho_4\rho_0^T=\rho_0$ & $\sigma_0\rho_1^T=\rho_0$ & $\sigma_0\rho_2\rho_2^T=\rho_3$ & $\rho_4\rho_3^T=\rho_3$\\ 
 \hline
 $\rho_{r_k}\overline{\rho}_{r_k}^T=$ & $\rho_0\rho_4^T=\rho_0$ & $\rho_1\sigma_0^T=\rho_0$ & $\rho_2\sigma_0^T=\rho_3$ & $\rho_3\rho_4^T=\rho_3$\\ 
 \hline
 $\rho_{r_k}\rho_{r_k}^T=$ & $\rho_0\rho_0^T=\rho_0$ & $\rho_1\rho_1^T=\rho_0$ & $\rho_2\rho_2^T=\rho_3$ & $\rho_3\rho_3^T=\rho_3$\\ 
 \hline
\end{tabular}
\end{center}
\caption{A table to show that $\overline{\rho}_{r_k} \rho_{r_k}^T=\rho_{r_k}\overline{\rho}_{r_k}^T=\rho_{r_k}\rho_{r_k}^T$.}
\label{table:ACompTable1}
\end{table}

Finally, by evaluating the expressions for $\mathcal{A}^c\mathcal{A}^{cT}$, $\mathcal{A}^c\mathcal{A}^T$, and $\mathcal{A}\mathcal{A}^{cT}$ into \eqref{eqn:UUT} we find that $UU^T=I$ as expected. Furthermore, $UU^T=I \implies U^T=U^{-1}$ since its inverse is unique, which proves that $U$ is unitary by definition.
\end{proof}


\section{Proof for Theorem \ref{U1U2}} \label{U1U2Proof}
\begin{proof}
    First, we show that $U=U_1U_2$. By expanding $U_1U_2$ out we can see that $U=U_1U_2$ if two conditions are met: $\mathcal{A}=\mathcal{A}\mathcal{A}^T\overline{\mathcal{A}}$ and $\mathcal{A}^c=(I-\mathcal{A}\mathcal{A}^T)\overline{\mathcal{A}}$. In Appendix \ref{ACompletionProof} it was shown that $\mathcal{A}\mathcal{A}^T=\mathcal{A}\overline{\mathcal{A}}^T$ and $\overline{\mathcal{A}}^T\overline{\mathcal{A}}=I$. Using these relations together gives $\mathcal{A}\mathcal{A}^T\overline{\mathcal{A}}=\mathcal{A}\overline{\mathcal{A}}^T\overline{\mathcal{A}}=\mathcal{A}I=\mathcal{A}$. From this, it follows that $(I-\mathcal{A}\mathcal{A}^T)\overline{\mathcal{A}}=\overline{\mathcal{A}}-\mathcal{A}\mathcal{A}^T\overline{\mathcal{A}}=\overline{\mathcal{A}}-\mathcal{A}=\mathcal{A}^c$, where the last step is given in Definition \ref{DefCompletionComplement}. Since both conditions are met, it holds that $U=U_1U_2$.

    Next, we show that both $U_1$ and $U_2$ are unitary. Since all of the matrices used in our particular decomposition are real, the conjugate transpose is equivalent to the transpose. Starting with $U_1$ we have
    \begin{equation} \notag
        U_1U_1^T = 
        \begin{pmatrix}
            (I-\mathcal{A}\mathcal{A}^T)^2 + \mathcal{A}\mathcal{A}^T\mathcal{A}\mathcal{A}^T & 2(\mathcal{A}\mathcal{A}^T-\mathcal{A}\mathcal{A}^T\mathcal{A}\mathcal{A}^T) \\
            2(\mathcal{A}\mathcal{A}^T-\mathcal{A}\mathcal{A}^T\mathcal{A}\mathcal{A}^T) & (I-\mathcal{A}\mathcal{A}^T)^2 + \mathcal{A}\mathcal{A}^T\mathcal{A}\mathcal{A}^T
        \end{pmatrix} \,.
    \end{equation}
    By \eqref{eqn:GeneralAejp1j} we have
    \begin{equation} \label{eqn:AAT}
        \mathcal{A}\mathcal{A}^T = (\bigotimes_{k=0}^{Q_1-1} \rho_{r_k}\rho_{r_k}^T)\otimes \mathcal{D} \otimes 
        I_{n_x}^{\otimes j-1} \,.
    \end{equation}
Using the fact that $(\rho_{r_k}\rho_{r_k}^T)\in \{\rho_0,\rho_3,\rho_4\}$ for $r_k\in \{0,\dots,4\}$, it follows that $\mathcal{A}\mathcal{A}^T$ is idempotent such that $\mathcal{A}\mathcal{A}^T\mathcal{A}\mathcal{A}^T=\mathcal{A}\mathcal{A}^T$. From this property it follows that $\mathcal{A}\mathcal{A}^T-\mathcal{A}\mathcal{A}^T\mathcal{A}\mathcal{A}^T=0$ and $(I-\mathcal{A}\mathcal{A}^T)^2+\mathcal{A}\mathcal{A}^T\mathcal{A}\mathcal{A}^T=I$ so $U_1U_1^T=I$. Since $U_1=U_1^T$ it follows that $U_1^TU_1=U_1U_1^T$ and therefore $U_1$ is unitary.

\noindent Finally, to show that $U_2$ is unitary we first note that $\overline{\mathcal{A}}$ as defined in \eqref{eqn:ABar} is unitary. Then we have
    \begin{equation} \notag
        U_2U_2^T = 
        \begin{pmatrix} \overline{\mathcal{A}} & 0 \\ 0 & \overline{\mathcal{A}} \end{pmatrix}
        \begin{pmatrix} \overline{\mathcal{A}}^T & 0 \\ 0 & \overline{\mathcal{A}}^T \end{pmatrix}
        =
        \begin{pmatrix} \overline{\mathcal{A}}\overline{\mathcal{A}}^T & 0 \\ 0 & \overline{\mathcal{A}}\overline{\mathcal{A}}^T \end{pmatrix} 
        = I \,,
    \end{equation}
where we once again use the property that $\overline{\mathcal{A}}\overline{\mathcal{A}}^T=I$. Since the inverse is unique, it follows that $U_2U_2^T=U_2^TU_2$ and, therefore, $U_2$ is also unitary. 
\end{proof}


\newpage
\twocolumngrid
\bibliography{bibliography.bib}

\end{document}